%% file: main.tex
\documentclass[sigconf]{acmart}

\pdfoutput=1

\newif\ifmaximal
\maximaltrue 

\ifmaximal
\else
	\usepackage{etoolbox}
	\newcommand{\zerodisplayskips}{%
	  \setlength{\abovedisplayskip}{0pt}%
	  \setlength{\belowdisplayskip}{0pt}%
	  \setlength{\abovedisplayshortskip}{0pt}%
	  \setlength{\belowdisplayshortskip}{0pt}}
	\appto{\normalsize}{\zerodisplayskips}
	\appto{\small}{\zerodisplayskips}
	\appto{\footnotesize}{\zerodisplayskips}
	\usepackage{stfloats}
\fi

\input{newcom}

\copyrightyear{2020}
\acmYear{2020}
\setcopyright{acmlicensed}
\acmConference[HSCC '20]{23rd ACM International Conference on Hybrid Systems: Computation and Control}{April 22--24, 2020}{Sydney, NSW, Australia}
\acmBooktitle{23rd ACM International Conference on Hybrid Systems: Computation and Control (HSCC '20), April 22--24, 2020, Sydney, NSW, Australia}
\acmPrice{15.00}
\acmDOI{10.1145/3365365.3382214}
\acmISBN{978-1-4503-7018-9/20/04}

\begin{document}

\title{Symbolic Controller Synthesis for B\"{u}chi Specifications\\ on Stochastic Systems}

\author{Rupak Majumdar} 
\orcid{https://orcid.org/0000-0003-2136-0542}
\affiliation{MPI-SWS, Germany.}
 \email{rupak@mpi-sws.org} 
 
 \author{Kaushik Mallik} 
 \orcid{https://orcid.org/0000-0001-9864-7475} 
 \affiliation{MPI-SWS, Germany.} 
 \email{kmallik@mpi-sws.org} 
 
 \author{Sadegh Soudjani} 
 \orcid{https://orcid.org/0000-0003-1922-6678} 
 \affiliation{Newcastle University, UK.} 
 \email{sadegh.soudjani@ncl.ac.uk}

\input{abstract}

\maketitle

\input{intro}

\input{prelims}

\input{abstraction}

\section{Controller synthesis and refinement}

\input{synthesis}
\input{abstraction_computation}
\input{examples}

\section{Future Work}
We are working on three different extensions of this work.
First, we plan to develop computation techniques for the qualitative winning regions for more general Rabin or parity conditions.
Second, we are working on formulating conditions to guarantee convergence of the computations to the actual winning region when the discretization gets finer.
Finally, we plan to improve the scalability of the approach using multi-resolution abstractions.
\begin{acks} 
	1. This research was funded in part by the \grantsponsor{dfg}{Deutsche Forschungsgemeinschaft}{https://www.dfg.de/} project \grantnum{dfg}{389792660-TRR 248}
and by the \grantsponsor{erc}{European Research Council}{https://erc.europa.eu/} under the
Grant Agreement \grantnum[http://www.impact-erc.eu/]{erc}{610150} (ERC Synergy Grant ImPACT).\\
	2. We thank Natsuki Urabe for giving insights on the possibility of trajectories being trapped in a bounded region even in the 
	presence of non-zero escape probability at each step.
\end{acks}

\bibliographystyle{abbrv}
\bibliography{reportbib}

\ifmaximal
	\input{appendix}
\fi

\end{document}

%% file: newcom.tex
\usepackage{comment}
\usepackage{amsmath}
\usepackage{amsthm}
\usepackage{etoolbox}
\usepackage{color,colortbl}
\usepackage{mathtools}
\usepackage{enumerate}
\usepackage{algorithm, algpseudocode}
\usepackage{booktabs}
\usepackage[normalem]{ulem}
\usepackage{pgf}
\usepackage{tikz,pgfplots}
\usetikzlibrary{trees,decorations,arrows,arrows.meta,automata,shadows,positioning,plotmarks,backgrounds,shapes,shapes.misc}
\usetikzlibrary{calc,matrix,fit,petri,decorations.pathmorphing,patterns}
\usetikzlibrary{decorations.pathreplacing,decorations.markings,shapes.geometric,calc}
\usepackage{paralist}
\usepackage{stmaryrd}
\usepackage{xspace}
\usepackage{graphicx}
\usepackage{float}
\usepackage[utf8]{inputenc} 
  \usepackage{csquotes} 
  \usepackage{algorithm}
 
\usepackage{multirow}
\usepackage{array}
\usepackage{dsfont}
 \usepackage{array}
\usepackage{xifthen}
\usepackage{relsize}

 \usepackage{subfig}

\newcommand{\KM}[1]{{\color{blue} \textbf{KM:} #1}}

\DeclareMathOperator{\Paths}{\mathsf{PATH}}

\newcommand{\Sadegh}{\textcolor{magenta}}

\newtheorem{problem}{Problem}
\newtheorem{assumption}{Assumption}

\newcommand{\Ker}{T_{\mathfrak s}}

\newtheorem{remark}{Remark}

\newcommand{\Den}{\mathcal{D}}

\newcommand{\Sys}{\mathfrak S}
\newcommand{\cell}[1]{\llbracket #1\rrbracket}
\newcommand{\Xs}{\mathcal S}
\newcommand{\xs}{x}
\newcommand{\Us}{\mathcal U}

\newcommand{\reach}{\ensuremath{\Phi}}

\newcommand{\Abs}{\mathcal{A}}
\newcommand{\Xh}{\widehat{\Xs}}
\newcommand{\xh}{\widehat{\xs}}
\newcommand{\Fo}{\overline{F}}
\newcommand{\Fu}{\underline{F}}

\newcommand{\Bh}{\underline{B}}
\newcommand{\Bho}{\overline{B}}

\newcommand{\cpre}{\mathit{Cpre}}
\newcommand{\pre}{\mathit{Pre}}
\newcommand{\apre}{\mathit{Apre}}
\newcommand{\upre}{\mathit{Upre}}

\newcommand{\Cont}{\mathcal{C}}
\newcommand{\Conth}{\widehat{\Cont}}

\newcommand{\W}{\mathsf{WinDom}}

\newcommand{\set}[1]{\lbrace #1 \rbrace}

\newcommand{\dom}{\mathsf{dom}\;}

\newcommand{\wu}{\underline{W}}
\newcommand{\wo}{\overline{W}}
\newcommand{\lo}{\overline{L}}

\newcommand{\REFthm}[1]{\text{Thm.~\ref{#1}}}
\newcommand{\REFdef}[1]{Def.~\ref{#1}}
\newcommand{\REFalg}[1]{Alg.~\ref{#1}}
\newcommand{\REFrem}[1]{Rem.~\ref{#1}}

\newcommand{\REFsec}[1]{Sec.~\ref{#1}}
\newcommand{\REFprop}[1]{Prop.~\ref{#1}}
\newcommand{\REFfig}[1]{Fig.~\ref{#1}}
\newcommand{\REFass}[1]{Assump.~\ref{#1}}

\newcommand{\REFtab}[1]{Table~\ref{#1}}
\newcommand{\REFex}[1]{Ex.~\ref{#1}}
\newcommand{\REFprob}[1]{Prob.~\ref{#1}}

%% file: abstract.tex

\begin{abstract}
We consider the policy synthesis problem for
continuous-state controlled Markov processes evolving in discrete time,
when the specification is given as a B\"uchi condition (visit a set of states infinitely
often).
We decompose computation of the maximal probability of satisfying the B\"uchi condition into two steps. The first step is to compute the maximal {\em qualitative winning set}, from where the B\"uchi condition can be
enforced with probability one. The second step is to find the maximal probability of reaching the already computed qualitative winning set.
In contrast with finite-state models, we show that such a computation only gives a lower bound on the maximal probability where the gap can be non-zero.

In this paper we focus on approximating the qualitative winning set, while pointing out that the existing approaches for unbounded
reachability computation can solve the second step.
We provide an abstraction-based technique to approximate the qualitative winning set by simultaneously using
 an over- and under-approximation
of the probabilistic transition relation.
Since we are interested in qualitative properties, the abstraction is non-probabilistic; instead, the probabilistic
transitions are assumed to be under the control of a (fair) adversary.
Thus, we reduce the original policy synthesis problem to a B\"uchi game under a fairness assumption and characterize
upper and lower bounds on winning sets as nested fixed point expressions in the $\mu$-calculus.
This characterization immediately provides a symbolic algorithm scheme.
Further, a winning strategy computed on the abstract game can be refined to a policy on the controlled Markov process.

We describe a concrete abstraction procedure and demonstrate our algorithm on two case studies.
We show that our techniques are able to provide tight approximations to the qualitative winning set for
the Van der Pol oscillator and a 3-d Dubins' vehicle.
\end{abstract}

%% file: intro.tex

\section{Introduction}

Decision making under stochastic uncertainty has many applications in science, engineering, and economics.
Typically, one models a system with uncertainty as a \emph{controlled Markov process} evolving in time. 
Such a process consists in a (possibly uncountable) set of states and actions.
At a given state, an agent picks an action and the state and the action together
determine the distribution over the next states.
The choice of the control action depends on the history of states seen so far and may be
randomized; the decision rule that assigns to each history a distribution of control actions
is called a \emph{policy}. 
Given a temporal specification over trajectories, the goal of the agent is to find an \emph{optimal} policy:
one that maximizes the probability that the resulting trajectory of the system satisfies the specification.
The control problem asks, given a controlled Markov process and a temporal specification (given, e.g.,
in linear temporal logic), to design an optimal policy.
In the finite-state setting, the control problem can be solved algorithmically 
based on graph traversal and linear programming 
\cite{BdA95,Luca98,BK08}.
A lot of recent research has focused on extending algorithmic policy synthesis techniques to \emph{continuous-state} systems.
The goal for continuous-state systems is to provide approximations to the probability of satisfaction, while
providing formal guarantees on convergence of the error.

While synthesis for reachability and safety properties have been studied in this setting both for infinite horizon \cite{TMKA17,HS_TAC19} and for finite horizon \cite{SSoudjani,LAB15,Kariotoglou17,LO17,Vinod18,LAVAE19,Jagtap2019}, 
there are few techniques for synthesis against \emph{B\"uchi} specifications, which requires the trajectory to visit
a given set of states infinitely often.
In this paper, we consider the problem of synthesizing controllers for controlled Markov processes 
for properties specified as B\"uchi conditions. 

The key aspect of the solution in the finite-state case is to separate a synthesis 
problem into a \emph{qualitative} part (find the set of states
from which the agent has a policy to satisfy the property almost surely) and 
a \emph{quantitative} part (find the policy that maximizes the probability of reaching the qualitatively winning states).
Given the qualitative solution, one can iteratively compute the quantitative solution by solving a reachability
problem, where the target is the absorbing set given by the qualitative solution.

Our first contribution is to show that a similar decomposition for B\"uchi properties does not hold for continuous state systems in general: we provide an example of a
Markov process over continuous state space for which the qualitative \emph{winning set} (from which there is a policy that ensures the B\"uchi property holds almost surely)
 is empty but the maximal probability of satisfying the property has a non-zero solution. 
Moreover, we show that such a decomposition is able to provide a lower bound on the quantitative part of the problem. 
Thus, if one can compute (an approximation of) the winning set, a 
lower bound on the quantitative solution can be obtained by Bellman iteration or by other 
techniques for (unbounded) reachability in the continuous-state setting \cite{TMKA17,HS_TAC19}. 


As our second contribution, we provide \emph{symbolic}
algorithms for computing under- and over-approximations of the qualitative winning set.
We compute finite-state abstractions of the continuous-state system. 
Our abstraction uses two transition relations: an over-approximation and an under-approximation of the continuous transitions.
For qualitative probabilistic analysis on finite-state systems, one can replace the probabilistic
transitions by an adversarial scheduler with a fairness requirement \cite{Pnueli83}.
Accordingly, our abstractions are \emph{non-probabilistic} and only require the knowledge of the support 
of the stochastic kernel associated to the process.
We characterize the qualitative winning states as a nested fixed point expression in the $\mu$-calculus \cite{Kozen83};
such an expression naturally gives a symbolic implementation. 
Since the abstraction is non-probabilistic, the symbolic implementation avoids numerical issues and can use standard
encodings based on satisfiability checkers or binary decision diagrams.

Our fixed point characterization is similar to the characterization of qualitative almost-sure winning in concurrent games \cite{dAH00},
but the use of two kinds of transitions---acting as upper and lower bounds---is a key distinguishing ingredient in our 
characterization.
We show through examples why both are required.
The qualitative winning region in the original continuous-state process is \emph{not} characterized 
by a similar fixed point: it is well known
that unlike the finite-state case, the actual probability values matter in deciding qualitative winning in infinite-state systems \cite[pp.~779-780]{BK08}.

We demonstrate our approach on the Van der Pol oscillator and the 3-d Dubin's vehicle, both in the presence of stochastic perturbation.  
Our computation shows that when the disturbance is treated as a worst case adversary, there exists a deterministic value 
of the disturbance for which the specification is violated for all initial states. 
On the contrary, when the disturbance is treated as random, we are able to satisfy the specification almost surely.
Moreover, we empirically show that the difference between the over- and under-approximation reduces as we pick finer abstractions.

%
There are relatively few results on the algorithmic analysis of liveness properties for controlled Markov processes.
A theoretical study of the B\"uchi objective $\square\lozenge B$ is conducted by Tkachev et al. \cite{TMKA17} 
via persistence properties $\lozenge\square \overline{B}$. 
It is shown that $P_s(\lozenge\square \overline{B})$ can be characterized by two fixed-point equations but no computational method is provided.
In particular, their techniques do not provide a way to solve the qualitative B\"uchi problem that we solve.
Our computational approach is similar in nature to results that employ Interval MDP or Interval Markov chains (MC) as abstractions. 
For example, Dutreix and Coogan \cite{Coogan19} use Interval MC for verification of a particular class of systems. 
The method of Dutreix and Coogan requires numerical computations of lower and upper bounds of the probabilities and provides an enumerative algorithm.
Our approach generalizes their construction of the over- and under-approximations to the setting of controlled Markov processes and the
synthesis problem. 
We focus on the winning region of the specification, which allows us to write symbolic algorithms purely on non-probabilistic structures, thus
avoiding numerical optimization procedures. 
Once the winning region is approximated, a quantitative reachability can be solved using standard techniques.

%% file: prelims.tex
\section{Controlled Markov Processes}


\subsection{Preliminaries}

We consider a probability space $(\Omega,\mathcal F_{\Omega},P_{\Omega})$,
where $\Omega$ is the sample space,
$\mathcal F_{\Omega}$ is a sigma-algebra on $\Omega$ comprising subsets of $\Omega$ as events,
and $P_{\Omega}$ is a probability measure that assigns probabilities to events.
We assume that random variables introduced in this article are measurable functions of the form $X:(\Omega,\mathcal F_{\Omega})\rightarrow (S_X,\mathcal F_X)$.
Any random variable $X$ induces a probability measure on  its space $(S_X,\mathcal F_X)$ as $Prob\{A\} = P_{\Omega}\{X^{-1}(A)\}$ for any $A\in \mathcal F_X$.
We often directly discuss the probability measure on $(S_X,\mathcal F_X)$ without explicitly mentioning the underlying probability space and the function $X$ itself.

A topological space $S$ is called a Borel space if it is homeomorphic to a Borel subset of a Polish space (i.e., a separable and completely metrizable space).
Examples of a Borel space are the Euclidean spaces $\mathbb R^n$, its Borel subsets endowed with a subspace topology, as well as hybrid spaces.
Any Borel space $S$ is assumed to be endowed with a Borel sigma-algebra, which is
denoted by $\mathcal B(S)$. We say that a map $f : S\rightarrow Y$ is measurable whenever it is Borel measurable.

We denote the set of nonnegative integers by $\mathbb N := \{0,1,2,\ldots\}$.


\subsection{Controlled Markov processes}
\label{sec:model}

We consider controlled Markov processes (CMP) in discrete time defined over a general state space, 
characterized by a tuple $\mathfrak S =\left(\mathcal S, \mathcal U, \Ker\right),$
where $\mathcal S$ is a Borel space serving as the state space of the process,
$\mathcal U$ is a finite input space, and 
$\Ker$ is a conditional stochastic kernel $\Ker:\mathcal B(\mathcal S)\times \mathcal S\times \mathcal U\rightarrow[0,1]$
with $\mathcal B (\mathcal S)$  being  the Borel sigma-algebra on the state space and $(\mathcal S, \mathcal B (\mathcal S))$ being the corresponding measurable space.
%
The kernel $\Ker$ assigns to any $s \in \mathcal S$ and $u\in\mathcal U$ a probability measure $\Ker(\cdot | s,u)$ 
on the measurable space $(\mathcal S,\mathcal B(\mathcal S))$
so that for any set $A \in \mathcal B(\mathcal S), P_{s,u}(A) = \int_A \Ker (ds|s,u)$, 
where $P_{s,u}$ denotes the conditional probability $P(\cdot|s,u)$.

\begin{remark}
The input space $\mathcal U$ in general can be any Borel space and the set of valid inputs can be state dependent. We have considered that $\mathcal U$ is a finite set and all elements of this set can be taken at any state. This choice is motivated by the digital implementation of control policies and also facilitates concise presentation of the results.
\end{remark}

\subsection{Semantics of controlled Markov processes}
The semantics of a CMP are characterized by its \emph{paths} or
executions, which reflect both the history of previous states of
the system and of implemented control inputs.
Paths are used to measure the performance of the system.

\begin{definition}
Given a CMP $\mathfrak S$, a \emph{finite path} is a sequence
\begin{equation*}
w_n = (s_0,u_0,\ldots, s_{n-1},u_{n-1},s_n),\quad n\in \mathbb N,
\end{equation*}
where $s_i\in\mathcal S$ are state coordinates and $u_i\in\mathcal U$ are control input coordinates of the path.
The space of all paths of length $n$ is denoted by $\Paths_n:=\mathcal K^n\times\mathcal S$ with $\mathcal K:=\mathcal S\times \mathcal U$.
Further, we denote projections by $w_n[i]:=s_i$ and $w_n(i):=u_i$.
An \emph{infinite path} of the CMP $\mathfrak S$ is the sequence
$w = (s_0,u_0,s_1,u_1,\ldots),$
where $s_i\in\mathcal S$ and $u_i\in\mathcal U$ for all $i\in\mathbb N$.
As above, let us introduce $w[i]:=s_i$ and $w(i):=u_i$.
The space of all infinite paths is denoted by $\Paths_\infty:=\mathcal K^\infty$.
\end{definition}

Given an infinite path $w$ or a finite path $w_n$, we assume below that $s_i$ and $u_i$
are their state and control coordinates respectively, unless otherwise
stated. For any infinite path $w\in \Paths_\infty$,
its $n$-prefix (ending in a state) $w_n$ is a finite path of length $n$, which we also call \emph{$n$-history}.
We are now ready to introduce the notion of control policy. 

\begin{definition}
\label{def:policy}
A \emph{policy} is a sequence $\rho = (\rho_0,\rho_1,\rho_2,\ldots)$ of universally measurable stochastic kernels $\rho_n$ \cite{BS96},
each defined on the input space $\mathcal U$ given $\Paths_n$.
The set of all policies is denoted by $\Pi$.
\end{definition}
Given a policy $\rho\in\Pi$ and a finite path $w_n\in \Paths_n$, the distribution of the next control input $u_n$ is given by
$\rho_n(\cdot|w_n)$.
In this work, we restrict our attention to the class of \emph{stationary} policies.
\begin{definition}
A policy $\rho$ is \emph{stationary} if there is a universally measurable function $\Cont:\mathcal S\rightarrow\mathcal U$ such that at any time epoch $n\in\mathbb N$, the input $u_n$ is taken to be $\Cont(s_n)\in\mathcal U$. Namely, the stochastic kernel $\rho_n(\cdot|w_n),\,n\in\mathbb N,$ in Definition~\ref{def:policy}  is a Dirac delta measure centered at $\Cont(s_n)$ with $s_n = w_n[n]$ being the last element of $w_n$.
We denote the class of stationary policies by $\Pi_{S}\subset\Pi$ and a stationary policy just by the function
$\Cont\in\Pi_S$. The function $\Cont$ is also called \emph{state feedback controller} in control theory. 
\end{definition}

For a CMP $\mathfrak S$, any policy $\rho\in \Pi$ together with an initial probability measure $\alpha:\mathcal{B}(\mathcal S)\rightarrow[0,1]$ of the CMP induces a unique probability measure on the canonical sample space of paths \cite{hll1996} denoted by $P_\alpha^\rho$ with the expectation $\mathbb E_\alpha^\rho$.
In the
case when the initial probability measure is supported on a single point,
i.e., $\alpha({s}) = 1$, we write $P_s^{\rho}$ and $\mathbb E_s^{\rho}$ in place of $P_\alpha^\rho$ and $\mathbb E_\alpha^\rho$, respectively.
We denote the set of probability measures on $(\mathcal S,\mathcal B(\mathcal S))$ by $\mathfrak D$.




\section{Problem Definition}
\label{sec:problem}
\noindent\textbf{Liveness specification.}
We consider liveness or repeated reachability specification as the synthesis objective.
Given a measurable set of states $B\subseteq \mathcal{S}$, 
the liveness specification is denoted by
$\square\lozenge B$ in linear temporal logic (LTL) notation \cite{BK08}.
An infinite path $w\in\Paths_\infty$ of a CMP $\mathfrak S$ satisfies the liveness specification $\square\lozenge B$  if for all $k_0\in \mathbb{N}$, there exists $k_1\in \mathbb{N}$ such that $k_1 > k_0$ and $w[k_1] \in B$.
This requires that the path $w$ visits the set $B$ infinitely often. 
We indicate the set of all infinite paths $w\in\Paths_\infty$ of $\mathfrak S$ that satisfy the property $\square\lozenge B$ by $\Sys\models \square\lozenge B$.

We are interested in the probability that the liveness specification can be satisfied by paths of a CMP $\mathfrak S$ under different policies. Given a policy $\rho\in\Pi$ and initial state $s\in\mathcal S$, we define the satisfaction probability as
\begin{equation}
\label{eq:sat_prob}
f(s,\rho) := P_s^{\rho}(\mathfrak S\models \square\lozenge B),
\end{equation}
and the supremum satisfaction probability
\begin{equation}
\label{eq:optimal_prob}
f^\ast(s) := \sup_{\rho\in\Pi} P_s^{\rho}(\mathfrak S\models \square\lozenge B).
\end{equation}

\begin{problem}[Policy Synthesis]
	\label{prob:non-trivial liveness}
	Given $\Sys$ and a set $B\subseteq \mathcal{S}$, find the optimal policy $\rho^*$ along with $f^\ast(s)$ s.t.\ $P_s^{\rho^*}(\mathfrak S\models \square\lozenge B) = f^\ast(s)$.
\end{problem}


Measurability of the event $\{\mathfrak S\models \square\lozenge B\}$ in the canonical sample space of paths under the probability measure $P_s^{\rho}$ is proved in \cite{TMKA17}. An initial attempt is also made to study the properties of the function $f^\ast(\cdot)$. For instance, it is shown that $f^\ast\equiv 1$ if and only if the probability that the path $w$ reaches $B$ is one for all initial states. 
We anticipate that the sets where $f^\ast(s)=1$ plays a crucial role in the computation of $f^\ast$.
\begin{definition}[Almost sure winning region]
Given the CMP $\mathfrak S$, the policy $\rho$, and a set $B\subseteq \mathcal{S}$, the state $s\in\mathcal S$ wins the specification $\square\lozenge B$ almost surely (a.s.\ in short) if $f(s,\rho)=1$. The a.s.\ \emph{winning region} of the policy $\rho$ is defined as
\begin{equation}
	  \W(\Sys,\rho):= \set{ s\in\mathcal S \mid f(s,\rho) = 1}.
\end{equation}
We also define the \emph{maximal a.s.\ winning region} as
\begin{equation}
	 \W^\ast(\Sys):= \set{ s\in\mathcal S \mid f^\ast(s) = 1}.
\end{equation}
\end{definition}

\begin{proposition}\label{prop:measurability of windom}
The set $\W^\ast(\Sys)$ is universally measurable. The set $\W(\Sys,\Cont)$ is also universally measurable for any stationary policy $\Cont\in\Pi_S$. The set $W := \W(\Sys,\Cont)$ is an absorbing set, i.e., the paths stating from this set will stay in the set with probability one.
\end{proposition}

\ifmaximal
	The proof can be found in the appendix.
\fi

In the sequel, we restrict our attention to stationary policies $\Cont\in\Pi_S$ and decompose the computation of $P_s^{\Cont}(\mathfrak S\models \square\lozenge B)$ into the computation of the winning set $\W(\Sys,\Cont)$ and then computation of reachability probability $P_s^{\Cont}(\mathfrak S\models \lozenge \W(\Sys,\Cont))$.
%
This is formalized next.

\begin{assumption}\label{ass:markov policy}
	There exists an optimal policy for \eqref{eq:optimal_prob}, and in addition that policy is stationary.
	Formally, there is a stationary policy $\Cont^\ast\in\Pi_S$ such that $f^\ast(s) = f(s,\Cont^\ast)$ for all $s\in \mathcal{S}$.
\end{assumption}

It is known, already for CMPs with countable state spaces, that for \REFprob{prob:non-trivial liveness} an 
optimal policy from a state $s\in \mathcal{S}$ may not exist even when $f^\ast(s)$ is non-zero \cite{kiefer2019b}.
As stated in the first part of \REFass{ass:markov policy}, we restrict the scope of the rest of the paper to problems for which an optimal policy exists.
For the second part of \REFass{ass:markov policy}, we do not yet know if stationary policies are sufficient for general CMPs even when an optimal policy is 
guaranteed to exist.
We do know that for CMPs with \emph{countable} state spaces, when an optimal policy exists, then it is stationary \cite{kiefer2017parity}.
If optimal policies need not be stationary,
our proposed algorithms in the following sections would still produce sound but sub-optimal approximations of the set 
$\W^\ast(\Sys)$; see \REFrem{rem:wrong assumption}.
Nevertheless, even if stationary policies are not sufficient for optimality, they are still of practical interest due to their simple structure and ease of implementation.
It is worthwhile to note that similar assumption was used by other authors \cite[Assump.~2]{TMKA17}.

\begin{theorem}
\label{thm:the problem can be broken down}
For any policy $\Cont\in\Pi_S$ on CMP $\Sys$, 
and $W := \W(\Sys,\Cont)$, we have
\begin{equation}
\label{eq:decom}
\begin{cases}
P_s^{\Cont}(\mathfrak S\models \square\lozenge B) = 1 & \text{if } s\in W\\
P_s^{\Cont}(\mathfrak S\models \square\lozenge B) \ge P_s^{\Cont}(\mathfrak S\models \lozenge W) & \text{if } s\notin W,
\end{cases}
\end{equation}
\end{theorem}

\ifmaximal
	The proof can be found in the appendix.
\fi



Computation of the reachability probability has been studied extensively in the literature for both infinite horizon \cite{TA12,tkachev2014characterization,TMKA17,HS_TAC19} and finite horizon \cite{SA13,SAH12,LAB15,Kariotoglou17,LO17,SAM17,Vinod18,LAVAE19,Jagtap2019} using different abstract models and computational methods. 
These approaches can be used to provide a lower bound on the probability of satisfaction of the B\"uchi condition.
So from this point onward, we mostly consider the first half of \eqref{eq:decom} which is formalized as in the following:

\begin{problem}[Maximal Winning Region]
\label{prob:policy}
Given $\Sys$ and a set $B\subseteq \mathcal{S}$, find a stationary policy $\Cont^\ast\in \Pi_S$ such that
$$\W(\Sys,\Cont)\subseteq \W(\Sys,\Cont^\ast) \text{ almost everywhere}$$
for all $\Cont\in\Pi_S$.
\end{problem}
The maximal policy $\Cont^\ast$ defined in Problem~\ref{prob:policy} is not necessarily unique, but the winning region associated to such maximal policies is unique.
\ifmaximal 
	A formal treatment of this claim can be found in the appendix in \REFsec{sec:unique winning region}.
\fi 
Moreover, when \REFass{ass:markov policy} holds, 
using properties of the sets $\W(\Sys,\Cont)$ for different policies, we have
$$\W(\Sys,\Cont^\ast) = \W^\ast(\Sys).$$


In the following, we mainly focus on the approximate computation of $\W(\Sys,\Cont^\ast)$ with a suitable policy.
%
For the reachability part, it has been shown that linking the infinite-horizon reachability to the finite-horizon one requires knowledge of the \emph{absorbing sets} from which the trajectory cannot escape.
So we briefly discuss in \REFsec{sec:minimal losing region} how the absorbing sets can be over-approximated to enable linking the reachability in \eqref{eq:decom} to its finite-horizon version.


\noindent\textbf{A solution outline for Problem~\ref{prob:policy}.}
Computing the exact maximal policy for $\Sys$ is difficult in general.
In the following we propose an approximation procedure which is motivated by the three-step abstraction-based controller synthesis methods \cite{ReissigWeberRungger_2017_FRR,tabuada09,nilsson2017augmented}:
\begin{description}
	\item[Abstraction.] First, the given CMP $\Sys$ is approximated using a finite state transition system $\Abs$, called the \emph{abstraction}, by means of state space discretization.
	The specification---which is the B\"uchi condition in our case---is also approximated using the discretized state space of the abstract transition system, and is called the \emph{abstract specification}.
	\item[Synthesis.] Second, the policy synthesis problem is posed as a zero-sum \emph{game} on $\Abs$ between the controller and an adversary, where in each step the controller chooses a control input, and the adversary chooses a successor allowed by the transitions in $\Abs$.
	The goal of the control player is a.s.\ satisfaction of the B\"uchi objective from as many discrete states as possible, whereas the goal of the environment player is the complement of the same.
	The outcome of the game, when played from the perspective of the control player, is an \emph{abstract controller} $\Conth$.
	\item[Controller refinement.] Third, the abstract controller $\Conth$ is mapped back to the continuous state space using a process called \emph{controller refinement}.
	This results in a continuous controller $\Cont$ that can be paired with $\Sys$ and a continuous winning domain $\W(\Sys,\Cont)$.
\end{description}

%% file: abstraction.tex

\section{A generic finite state abstraction}\label{sec:abstraction}

Our proposed solution relies on constructing an abstraction $\Abs$ which uses two transition functions to approximate the transition kernel of $\Sys$.
\begin{definition}
A transition system is a tuple $(Q,\Sigma,\delta_1,\delta_2)$ where $Q$ is a finite set of states, $\Sigma$ is a finite input alphabet, and $\delta_i:Q\times \Sigma \rightarrow 2^Q$, $i\in\set{1,2}$ are two transition functions with the property that $\delta_1(q,u)\supseteq \delta_2(q,u)$ for all $(q,u)\in Q\times \Sigma$.
\end{definition}

The abstraction $\Abs$ is constructed based on a finite partition of the state space. Therefore, we require a bounded state space $\Xs$. 
If $\Xs$ is unbounded, we truncate it to a measurable set $\Xs'$, which serves as the working region; the rest is represented by a symbolic \emph{sink} state $\phi$.
It is true that this truncation could possibly lead to computation of only an under-approximation of the true a.s.\ winning region.
However, from practical standpoint, it is often the case that bounds of the state variables are already in place (e.g.\ the work space of a robot, or the maximum rated voltage across a capacitor in a voltage converter, etc.), which would serve as the boundary of $\Xs'$.
The state $\phi$ models all the out-of-domain behaviors of $\Sys$, and should be avoided.
The new CMP will be $\mathfrak S' := (\mathcal S'\cup\set{\phi},\mathcal U,\Ker')$ with
\begin{equation}
\Ker'(A\,|\,s,u) = \begin{cases}
\Ker(A\,|\,s,u) & \text{if } s\in \mathcal S'\\
0 & \text{if }  s = \phi,
\end{cases}
\end{equation}
for any $A\in\mathcal B(\mathcal S')$, and $\Ker'(\phi\,|\,s,u) = 1-\Ker'(\mathcal S'\,|\,s,u)$.
In order to avoid change of notation, we work in the sequel with $\mathfrak S := (\Xs,\mathcal U,\Ker)$ where the state space $\Xs$ is bounded but may also include a symbolic sink state $\phi$. We also assume that $B\subseteq \mathcal{S}$.

\subsection{The abstraction}

We propose a novel abstraction, which will later be used in Sec.~\ref{sec:synthesis} to synthesize approximations of the maximal winning region and a winning policy.
First, we introduce some notation.
Given the state space $\Xs$ of the CMP $\Sys$, we define a finite partition of $\Xs$ denoted by
 $\Xh := \set{\xh_i}_{i\in I}$ s.t.\ $\Xs = \bigcup_{i\in I} \xh_i $ and $\xh_i\cap \xh_j=\emptyset$ for all $\xh_i,\xh_j\in \Xh$, $i\neq j$.
The set $\Xh$ will be called the abstract state space (state space of the abstraction), and each element $\xh_i$ is an abstract state.
\begin{remark}
For the theory that is going to be presented in this paper, the abstract states need not be of the same size.
However, for a practical implementation, partition sets are chosen to be hyper-rectangular of the form 
$\xh = \cell{a,b}$ where $a,b\in \Xs$ are vectors. The partition sets are uniformly sized and their boundaries are assigned to only one partition element. 
In our implementation, we have also used a hyper-rectangular state space with an additional symbolic state $\phi$ that is also an element of the abstract state space.
\end{remark}

\begin{definition}\label{def:abstraction}
	Let $\Sys = (\Xs,\Us,\Ker)$ be a CMP and $\Xh$ be a finite partition of $\Xs$.
	Suppose $\Abs = (\Xh,\Us,\Fo,\Fu)$ is a transition system 
	with two modes of transitions $\Fo:\Xh\times \Us\rightarrow 2^{\Xh}$ and $\Fu:\Xh\times \Us\rightarrow 2^{\Xh}$.
	The system $\Abs$ will be called an \emph{abstraction} of $\Sys$ if for all $\xh\in \Xh$ and all $u\in \Us$,
	\begin{align}
		\Fo(\xh,u) &\supseteq \set{\xh'\in\Xh \mid \exists s\in \xh\;.\;\Ker(\xh'\mid s,u) > 0},\label{eq:def fo}\\
		\Fu(\xh,u) &\subseteq \set{\xh'\in\Xh\mid \exists \varepsilon>0\;.\; \forall s\in \xh\;.\;\Ker(\xh'\mid s,u) \geq \varepsilon}.\label{eq:def fu}
	\end{align}
\end{definition}



In words, in the presence of the stochastic disturbance and given an abstract state $\xh$ and a control input $u\in \Us$, $\Fo(\xh,u)$ represents an over-approximation of the set of all abstract states which can be reached with positive probability from \emph{some} continuous state in $\xh$, and $\Fu(\xh,u)$ represents a subset of $\Fo(\xh,u)$ under-approximating the set of states that can be reached with probability bounded away from $0$ from \emph{all} the states in $\xh$.
The role of the lower bound $\varepsilon$ in \eqref{eq:def fu} will be clear in \REFsec{sec:a.s. progress}.
We defer the actual computation of the abstract transition system until Sec.~\ref{sec:abs computation}.

\begin{remark}\label{rem:two-adversaries}
	The abstract transition systems in the usual abstraction based control methods \cite{ReissigWeberRungger_2017_FRR,tabuada09,nilsson2017augmented} play the role of a game graph for a two-player zerosum game between the controller and an imaginary adversary, where the adversary is an accumulation of the external perturbation and the discretization-induced non-determinism.
	The rule of the game is that at each discrete step, the controller plays a control input, to which the adversary responds by choosing one of the many non-deterministic successors.
	A control policy is synthesized for the controller by treating the adversary actions in worst-case fashion.

	In our case, the controller effectively plays simultaneously against \emph{two} imaginary adversaries who use two different types of actions: 
	The first adversary---called the random adversary---uses the external random noise, while the second adversary---called the non-deterministic adversary---uses the discretization-induced non-determinism.
	This separation enables the controller to somewhat relax the worst-case treatment of the problem, by assuming that the random adversary 
	is fair in choosing its actions, meaning all the noise values in the support of the distribution will appear always eventually.
	The non-deterministic adversary is still treated in the worst-case fashion.
	
	Keeping this two-adversary interpretation in perspective, given some control action, one can interpret a $\Fo$-transition as a joint colluding move of the two adversaries, while one can interpret a $\Fu$-transition as a move of the lone random adversary.
	The rule of the game in our case is that at each discrete step, the controller plays a control input, to which either the adversaries jointly respond by choosing one of the  $\Fo$-successors (which is possibly not an $\Fu$-successor), or the random adversary independently chooses an $\Fu$-successor.
	Since the random adversary is fair in its moves, hence it will not collude with the non-deterministic adversary all the time, and all the $\Fu$-transitions will be chosen at some point in the long run.
	This additional fairness assumption in the underlying game creates a favorable condition for the controller in many cases, as will be shown in the next section.
\end{remark}


\subsection{Almost sure progress}\label{sec:a.s. progress}
The fairness of the random adversary is materialized using the $\varepsilon$ in the definition of $\Fu$, which guarantees that a trajectory eventually exits from an abstract state $\xh$ in the long run, even when there is a non-zero probability for a single-step successor of a continuous state $s\in \xh$ to stay within $\xh$.
This feature is a central element of our synthesis method that will be presented in \REFsec{sec:synthesis}.

Following is an example of a general continuous state Markov chain which demonstrates that in the absence of the bound $\varepsilon$ in the definition of $\Fu$, trajectories could get trapped inside $\xh$ forever.

\begin{example}
\label{ex:trapped transitions}
	For simplicity and a focused exposition of the actual issue, we use a system admitting no control over state trajectories (can be alternatively thought of as a system with a control input space that has a single element).
	Consider a one dimensional CMP with state space $\Xs = [0,2]$, and with the following transition kernel (see \REFfig{fig:trapped trajectory}): 
	\begin{align}
	&s_k\in [1,2] \Rightarrow\;\Ker([\alpha,\beta]\mid s_k) = (\beta-\alpha),\nonumber\\
	&s_k=0 \Rightarrow\;\begin{cases}\label{eq:ex_trapped}
									&\Ker\left(\set{0}\mid s_k\right) = 0.5\\
									&\Ker\left([\alpha,\beta]\mid s_k\right) = 0.5(\beta-\alpha),
								\end{cases}\\
		&s_k\in (0,1) \Rightarrow\;\begin{cases}
										&\Ker\left(\set{b(s_k)} \;\vert\; s_k\right) = 1-a(s_k)\\
										&\Ker\left([\alpha,\beta]\mid s_k\right) = a(s_k)(\beta-\alpha),
									\end{cases}\nonumber
	\end{align}
	for any $[\alpha,\beta]\subset [1,2]$ with $\alpha\le\beta$, and $a:s_k \mapsto [0,1]$ is some probability assigning function, and  $b(s_k) := \frac{s_k}{1+s_k}$. In words, the next state is uniformly distributed over the interval $[1,2]$ if the current state of the CMP is in the same interval; for current state $s_k=0$, the next state either stays at zero with probability $0.5$ or jumps uniformly to the interval $[1,2]$; for current state $s_k\in(0,1)$, the next state jumps to the interval $[1,2]$ with probability $a(s_k)$ or jumps to a single state $b(s_k)$ with probability $1-a(s_k)$.	
	\begin{figure}
	{\scriptsize
		\begin{tikzpicture}[scale=4,->,>=stealth',shorten >=1pt,auto,semithick,every state/.style={semithick,fill=black,scale=0.15}]
			\draw[-]	(0,0)	--	(2,0);
			\draw[-]	(0.05,-0.1)	--	(0,-0.1)	--	(0,0.1)	--	(0.05,0.1);
			\draw[-]	(2-0.05,-0.1)	--	(2,-0.1)	--	(2,0.1)	--	(2-0.05,0.1);
			\draw[-]	(1,-0.1)	--	(1,0.1);
			
			\node[state,label=left:{\color{red} $0$}]	(A)	at	(0,0)	{};
			\node[state,label=above right:{\color{red} $\frac{1}{2}$}]	(B)	at	(0.5,0)	{};
			\node[state,label=below:{\color{red} $\frac{1}{3}$}]	(C)	at	(0.333,0)	{};
			\node[state,label=below:{\color{red} $\frac{1}{4}$}]	(D)	at	(0.25,0)	{};
			\node[state,label=below:{\color{red} $\frac{1}{5}$}]	(E)	at	(0.2,0)		{};
			\node[state,label=below:{\color{red} $\frac{1}{6}$}]	(F)	at	(0.166,0)	{};
			\node[state,,label=below:{\color{red} $\frac{1}{7}$}]	(G)	at	(0.142,0)	{};
			
			\node[state,label=below left:{\color{red} $1$}]	(H)	at	(1,0)	{};
			\node			(G)	at	(1.5,0)	{};
			\node[label=below left:{\color{red} $2$}]	at	(2,0)	{};
			
			\path	(A)	edge[loop above]		node	{$0.5$}	()
					(A)	edge[bend left]		node	{$0.5$}	(G)
					(B)	edge[bend left]		node	{$0.75$}	(C)
						edge[bend right]		node[below]	{$0.25$}	(G)
					(C)	edge[bend right]	node[above]	{$0.889$}	(D)
						edge[bend left]	node[below]	{$0.111$}	(G)
					(H)	edge[bend left]	node[below]	{$1$}	(G)
					(G)	edge[loop above]	node[above]	{$1$}	();
		\end{tikzpicture}
	}
	\caption{A trajectory could get trapped inside $[0,1)$ for infinite time if the transition probabilities are allowed to be arbitrarily small.
	In the figure, the nodes with red labels are continuous states, and the labels on the edges are the probability of the associated transition.}
	\label{fig:trapped trajectory}
	\end{figure}		
	
	Let us consider two CMPs $\mathfrak S_1,\mathfrak S_2$ with kernels obtained respectively by $a(s) = s^2$ and $a(s)=0.5$ for all $s\in(0,1)$, and compute the probability $P_s(\mathfrak S_i\models \square\lozenge [0,1))$. For a trajectory starting from an initial state $s_0\in (0,1)$, the probability of staying inside $[0,1)$ is:
	\begin{align*}
		&P_{s_0}(\mathfrak S_1\models \square[0,1))\\
		=&\left( 1-s_0^2 \right)\cdot \left( 1-\left(\frac{s_0}{1+s_0}\right)^2 \right)\cdot \left( 1-\left(\frac{\frac{s_0}{1+s_0}}{1+\frac{s_0}{1+s_0}}\right)^2 \right) \ldots \\
		 = &\left( 1-s_0^2 \right)\cdot \left( 1-\frac{s_0^2}{(1+s_0)^2} \right)\cdot \left( 1-\frac{s_0^2}{(1+2s_0)^2} \right) \ldots \\
		= &\prod_{k=0}^\infty \left( 1-\frac{s_0^2}{(1+ks_0)^2} \right) \\
		= &1-s_0.
	\end{align*}
	Then, for any $s_0\in (0,1)$ there is a non-zero probability of staying forever inside $[0,1)$.
	For example, for $s_0=0.5$, this probability is $0.5$. Doing the computations for other states results in
	\begin{equation*}
	P_s(\mathfrak S_1\models \square\lozenge [0,1)) = 
	\begin{cases}
	1-s & \text{if } s\in(0,1)\\
	0 & \text{if } s\in[1,2]\cup\set{0}.
	\end{cases}
	\end{equation*}
For the second model with $a(s)=0.5$, we have $P_{s}(\mathfrak S_2\models \square\lozenge[0,1))=0$ for all $s\in[0,2]$.
\qed
\end{example}

\REFex{ex:trapped transitions} clearly shows that unlike discrete MDPs, in case of continuous-space CMPs we cannot ignore the actual value of the probabilities, as otherwise we would have had $P_{s}(\mathfrak S_i\models \square[0,1))=0$ for both $i\in\set{1,2}$. 
So unlike the case of discrete MDPs, we can no longer just use the support of the distribution to find the winning region, and then solve a reachability problem. 

\REFex{ex:trapped transitions} also shows that the inequality in \eqref{eq:decom} can be strict: the winning region is empty but there are states with positive probability of satisfying the liveness specification. We leave the formulation of conditions under which the equality holds as future work.

The CMP $\mathfrak S_1$ in \REFex{ex:trapped transitions} also justifies the use of $\varepsilon$ in the definition of $\Fu$. Assume that we want to compute an abstraction using the hyper-rectangular cover $\set{[0,1),[1,2]}$. The continuous states in the cell $[0,1)$ have positive transition probability to $[1,2]$, although there \emph{does not exist} a uniform lower bound $\varepsilon>0$ of these transition probabilities. 
As we showed in \REFex{ex:trapped transitions}, even then a trajectory can remain trapped inside $[0,1)$ with a positive probability in the long run. Because of the presence of the $\varepsilon$ in the definition of $\Fu$, the abstraction will have $\Fu([0,1)) = \emptyset$ (there is no control input, so $\Fu$ takes only state as argument).

\begin{proposition}\label{prop:finite escape}
	Let $\xh\in \Xh$ be an abstract state and $u\in U$ be an input s.t.\ there exists $\xh'\in \Xh$, $\xh' \neq \xh$, such that $\xh'\in\Fu(\xh,u)$.
	Then, there is a policy $\Cont$ such that 
	$$P_s^{\Cont}(\mathfrak S \models \lozenge\square \xh) = P_s^{\Cont}(\mathfrak S \models \square \xh) = 0,\quad\forall s\in \Xs,$$
	i.e.\ the probability of an infinite trajectory getting trapped forever inside $\xh$ is $0$.
\end{proposition}

\ifmaximal
	\begin{proof}
	Since $\xh'\in\Fu(\xh,u)$, there exists an $\varepsilon>0$ such that $\Ker(\xh'\mid s,u) \geq \varepsilon$ for all $s\in\xh$. Define the policy to be the constant one $\Cont(s) = u$ for all $s\in\xh$ and any other input policy for other states $s\notin\xh$. We first show that $P_s^{\Cont}(\mathfrak S \models \square \xh) = 0$.
	Note that $P_s^{\Cont}(\mathfrak S \models \square \xh) = \lim_{n\rightarrow\infty} V_n(s)$, with $V_0(s) = \mathbf 1_{\xh}(s)$ and $V_{n+1}(s) = \int_{\xh}V_n(s')\Ker(ds'|s,u)$. It is easy to show inductively that $V_n(s)\le (1-\varepsilon)^n$ for all $s$. By taking the limit as $n\rightarrow\infty$ we get the claimed result. For the $\lozenge\square \xh$ property, we have $P_s^{\Cont}(\mathfrak S \models \lozenge\square \xh) = \lim_{n\rightarrow\infty} W_n(s)$ with $$W_{n+1}(s) = P_s^{\Cont}(\mathfrak S \models \square \xh) +\int_{\Xs\backslash\xh}W_n(s')\Ker(ds'|s,u),$$ which results in $W_n(s) = 0$ for all $n$.   
	\end{proof}
\fi

Intuitively, \REFprop{prop:finite escape} says that if there is a $\xh'\in \Fu(\xh,u)$ with $\xh'\neq\xh$, then almost all trajectories reaching $\xh$ will eventually leave $\xh$ in finite time with the repeated use of the control input $u$.

We presented \REFex{ex:trapped transitions} to show that having $\Ker(\xh'\mid s,u)>0$ for all $s\in \xh$ is inadequate, and we need a uniform lower bound $\varepsilon$ for the definition of $\Fu$ in \eqref{eq:def fu}.
In fact, this condition is enough under proper continuity assumptions on the stochastic kernel.
\begin{proposition}
\label{prop:cont}
Suppose the kernel $\Ker$ is continuous, i.e., $g(s,u) = \int_{\Xs} f(s')\Ker(ds'\mid s,u)$ is upper (lower) semi-continuous for any upper (lower) semicontinuous function $f$. Moreover, the partition sets in $\Xh = \set{\xh_i}_{i\in I}$ are such that $\Ker(\mathsf{int}(\xh')\mid s,u) = \Ker(\mathsf{cl}(\xh')\mid s,u)$ for all $s\in\mathsf{cl}(\xh)$, $u\in\mathcal U$, and $\xh,\xh'\in\Xh$, where $\mathsf{int}(\cdot)$ and $\mathsf{cl}(\cdot)$ indicate respectively the interior and the closure of a set. Then, $\Fu(\xh,u)$ can be defined alternatively as
\begin{equation}
\label{eq:under_alter}
		\Fu(\xh,u)\subseteq \set{\xh'\in\Xh\mid\; \forall s\in \mathsf{cl}(\xh)\;.\;\Ker(\xh'\mid s,u)>0}.
	\end{equation}
\end{proposition}

\ifmaximal
	The proof can be found in the appendix.
\fi

%% file: synthesis.tex

\label{sec:synthesis}
With the abstraction $\Abs$ of the given CMP $\Sys$ computed in \REFsec{sec:abstraction}, we now propose algorithms to approximate---from above and below---the maximal a.s.\ winning region $\W(\Sys,\Cont^\ast)$.
As a by-product, we will also obtain a suitable control policy.
We first lift the specification $\square \lozenge B$ to an abstract specification that can be specified using the states of $\Abs$.
For that we define an under-approximation $\Bh$ and an over-approximation $\Bho$ of the set $B\subseteq \mathcal{S}$ using the states of $\Abs$:
\begin{align*}
	\Bh := \set{\xh\in \Xh \mid \ \forall s\in \xh\;.\;B(s)=1}\\
	\Bho := \set{\xh\in \Xh \mid \ \exists s\in \xh\;.\;B(s)=1}.
\end{align*}

Note that the sink state $\phi\notin \Bh$, since we assume that the set $B$ is fully contained in the working region of the CMP. Hence, the satisfaction of $\square\lozenge \Bh$ would ensure that $\phi$ is always avoided.

To formalize the synthesis process, we introduce four operators:

\begin{description}
	\item[Controllable predecessor:] Define $\cpre_F:2^{\Xh}\rightarrow 2^{\Xh}$ for $F\in \set{\Fo,\Fu}$, 
	\begin{equation}\label{eq:cpre}
		\cpre_F:T\mapsto \set{\xh\in \Xh \mid \exists u\in \Us\;.\; F(\xh,u)\subseteq T}.
	\end{equation}
	\item[Cooperative predecessor:] Define $\pre_F:2^{\Xh}\rightarrow 2^{\Xh}$ for $F\in \set{\Fo,\Fu}$,
	\begin{equation}
		\pre_F:T\mapsto \set{\xh\in \Xh \mid \exists u\in \Us\;.\; F(\xh,u)\cap T \neq \emptyset}.
	\end{equation}
	\item[Almost sure predecessor:] Define $\apre:2^{\Xh}\times 2^{\Xh}\rightarrow 2^{\Xh}$,
	\begin{equation}
		\apre:(T,S)\mapsto \set{\xh\in \Xh \mid \exists u\in \Us\;.\; \Fo(\xh,u)\subseteq T \wedge \Fu(\xh,u)\cap S \neq \emptyset}.
	\end{equation}
	\item[Uncertain predecessor:] Define $\upre:2^{\Xh}\times 2^{\Xh}\rightarrow 2^{\Xh}$,
	\begin{equation}\label{eq:upre}
		\upre:(T,S)\mapsto \set{\xh\in \Xh\mid \exists u\in \Us\;.\; \Fu(\xh,u)\subseteq T \wedge \Fo(\xh,u)\cap S \neq \emptyset}.
	\end{equation}
\end{description}


\subsection{Warm-up: reachability specification}
As a warm-up, we first consider under-approximation of the largest winning domain for a.s.\ satisfaction of the reachability specification $\lozenge B$ in the presence of stochastic noise.
For that we take inspiration from the fixed point of a.s.\ reachability in concurrent two-player games \cite{dAH00}, and obtain the following nested fixed-point on the abstract system $\Abs$:
\begin{align}\label{eq:a.s. reach fixed-point}
	\nu Y\;.\;\mu Z\;.\; (\Bh^c\cap (\apre(Y,Z)\cup \cpre_{\Fo}(Z))) \cup \Bh,
\end{align} 
where $\Bh^c$ represents the complement of the set $\Bh$.
Intuitively, the above fixed point computes the largest set $Y$ s.t.\ from every state $y\in Y\setminus \Bh$ there exists control input sequence s.t.\
either
 \begin{inparaenum}[(a)]
	\item there is a finite sequence of $\Fu$-transitions to $\Bh$ and no $\Fo$-transition outside $Y$, or 
	\item there is a bounded finite sequence of $\Fo$-transitions all of whose non-deterministic branches reach $\Bh$.
\end{inparaenum}
On the CMP level this means: for all  $y\in Y\setminus \Bh$ and for all $s\in y$ there exist control input sequence s.t.\ either 
\begin{inparaenum}[(a)]
	\item there exist paths that enter $\Bh$ with positive probability bounded away from zero and all paths stay inside $Y$ with probability $1$, or 
	\item there exist paths that enter $\Bh$ with probability $1$.
\end{inparaenum}
It can be shown, by repeated use of \REFprop{prop:finite escape} on the $\Fu$-transitions, that $\Bh$ (and hence $B$) will be reached a.s.\ by such a control sequence.

The fixed point \eqref{eq:a.s. reach fixed-point} is in contrast with the usual reachability fixed point for worst case disturbances, which is given as $\mu Z\;.\;\cpre_{\Fo}(Z)\cup \Bh$, where it is required that from every state $z\in Z$, all the non-deterministic branches of $\Fo$ reach $\Bh$ in at most some finite number of steps. 
Note that the solution of \eqref{eq:a.s. reach fixed-point} subsumes the solution of the usual fixed point, and in practice the usual reachability fixed point is much stronger than its stochastic counterpart.
Following is an illustrative example that captures this intuition.

\begin{example}\label{ex:soft treatment of liveness}
	Consider the CMP $\mathfrak S_2$ (see \REFfig{fig:cpre is too strong}) defined in \REFex{ex:trapped transitions}, with stochastic kernel \eqref{eq:ex_trapped} and $a(s)=0.5$.
	The transition probabilities from states in $[0,1)$ to the interval $[1,2]$ have a lower bound $0.5$, and so unlike $\mathfrak S_1$ of \REFex{ex:trapped transitions}, $\Fu([0,1)) = \set{[0,1),[1,2]}$.
	Moreover, since there exist transitions with positive probability from all the state in $[0,1)$ to $[1,2]$, hence $\Fo([0,1)) = \set{[0,1),[1,2]}$ as well.

	Assume that the abstract reachability specification is given $\lozenge [1,2]$.
	If we start from some state in $[0,1)$ and treat the adversary (resolving the non-determinism) as worst-case, then by using both $\Fu$ or $\Fo$, we will forever loop in $[0,1)$ and would violate the specification.
	Formally, the fixed point for $\lozenge [1,2]$ will converge to the singleton set $\set{[1,2]}$, since $\cpre_{\Fo}(\set{[1,2]}) = \set{[1,2]}$.
	
	On the other hand, if we treat the disturbance as stochastic noise, then from \REFprop{prop:finite escape} we know that if we loop in $[0,1)$ indefinitely long, then in the long run a.s.\ the system is going to move to $[1,2]$ (recall the interpretation of fair random adversary).
	So the winning region in this case should be the whole state space $\set{[0,1),[1,2]}$.
	Indeed, \eqref{eq:a.s. reach fixed-point} will include the state $[0,1)$ since $[0,1)\in \apre(\set{[0,1),[1,2]},\set{[1,2]})$.
%
	\begin{figure}
	{\scriptsize
		\begin{tikzpicture}[scale=4,->,>=stealth',shorten >=1pt,auto,semithick]
			\draw[-]	(0,0)	--	(2,0);
			\draw[-]	(0.05,-0.1)	--	(0,-0.1)	--	(0,0.1)	--	(0.05,0.1);
			\draw[-]	(2-0.05,-0.1)	--	(2,-0.1)	--	(2,0.1)	--	(2-0.05,0.1);
			\draw[-]	(1,-0.1)	--	(1,0.1);
			
			\node[state,label=left:{\color{red} $0$},fill=black,scale=0.15]	(A)	at	(0,0)	{};
			\node[state,label=above right:{\color{red} $\frac{1}{2}$},fill=black,scale=0.15]	(B)	at	(0.5,0)	{};
			\node[state,label=below:{\color{red} $\frac{1}{3}$},fill=black,scale=0.15]	(C)	at	(0.333,0)	{};
			\node[state,label=below:{\color{red} $\frac{1}{4}$},fill=black,scale=0.15]	(D)	at	(0.25,0)	{};
			\node[state,label=below:{\color{red} $\frac{1}{5}$},fill=black,scale=0.15]	(E)	at	(0.2,0)		{};
			\node[state,label=below:{\color{red} $\frac{1}{6}$},fill=black,scale=0.15]	(F)	at	(0.166,0)	{};
			\node[state,,label=below:{\color{red} $\frac{1}{7}$},fill=black,scale=0.15]	(G)	at	(0.142,0)	{};
			
			\node[state,label=below left:{\color{red} $1$},fill=black,scale=0.15]	(H)	at	(1,0)	{};
			\node			(G)	at	(1.5,0)	{};
			\node[label=below left:{\color{red} $2$}]	at	(2,0)	{};
			
			\path	(A)	edge[loop above]		node	{$0.5$}	()
					(A)	edge[bend left]		node	{$0.5$}	(G)
					(B)	edge[bend left]		node	{$0.5$}	(C)
						edge[bend right]		node[below]	{$0.5$}	(G)
					(C)	edge[bend right]	node[above]	{$0.5$}	(D)
						edge[bend left]	node[below]	{$0.5$}	(G)
					(H)	edge[bend left]	node[below]	{$1$}	(G)
					(G)	edge[loop above]	node[above]	{$1$}	(G);
					
			 \node[draw, single arrow,
              minimum height=10mm, minimum width=8mm,
              single arrow head extend=2mm,
              anchor=west, rotate=-90] at (1,-0.3) {};
              \node		at		(1.25,-0.4)		{\normalsize Abstraction};
              
             \node[state]	(I)	at	(0.5,-0.7)	{\color{red} $[0,1)$};
             \node[state]	(J)	at	(1.5,-0.7)	{\color{red} $[1,2]$};
             \path	(I)	edge[loop left]	node	{}	()
             			edge			node	{}	(J)
             		(J)	edge[loop right]	node	{}	();
             			
             \node	at	(1,-0.8)	{\normalsize $\Fo(\cdot) \equiv \Fu(\cdot)$};
		\end{tikzpicture}
	}
	\caption{The top figure represents $\Sys_2$ and the bottom figure represents both transition relations $\Fu([0,1))$ and $\Fo([0,1))$ (same in this case) of $\Abs$.
	It is not possible to reach $[1,2]$ from $[0,1)$ if the non-determinism in $\Fu$ or $\Fo$ is treated as a worst-case adversary: the adversary can choose the loop at $[0,1)$ all the time.
	In the top figure, the labels on the edges are the probabilities of the associated transitions, and in both figures the labels in red are continuous states.}
	\label{fig:cpre is too strong}
	\end{figure}	
	\qed	
\end{example}

\begin{remark}\label{rem:progress group}
	Nilsson et al.\ \cite{nilsson2017augmented} introduced augmented transition systems as abstractions of non-stochastic systems.
	Augmented transition systems embed liveness information in progress groups.
	If a set of abstract states form a progress group under some control action, then the system eventually leaves the progress group under 
	repeated use of this particular control action.
	Even though our work deals with an unrelated problem, we note that our fairness assumption on the random adversary allowing the 
	CMP to make progress to escape a given abstract state a.s.\ (\REFprop{prop:finite escape}) has a similar flavor.
\end{remark}

\subsection{Under-approximation of the maximal a.s.\ winning region}
\label{sec:under_app}
We build up on the intuition of the solution of the a.s.\ reachability specification, and present the computation of a sound under-approximation of the maximal a.s.\ winning region $\W(\Sys,\Cont^\ast)$ with a suitable abstract controller $\Conth$.
In $\mu$-calculus notation, this under-approximation can be computed as:
\begin{multline}\label{eq:wu}
	\wu := \nu Y\;.\;\mu Z\;.\;\left[ (\Bh^c \cap (\apre(Y,Z) \cup \cpre_{\Fo}(Z)))\right. \\
							\left. \cup (\Bh\cap \cpre_{\Fo}(Y)) \right].
\end{multline}
Note that the only new term in \eqref{eq:wu} as compared to \eqref{eq:a.s. reach fixed-point} is the intersection of $\Bh$ with $\cpre_{\Fo}(Y)$.
This additional term makes sure that each time $\Bh$ is reached, the winning region $Y$ is not left in the next step to make sure that $\Bh$ can be reached once again.

The fixed point \eqref{eq:wu} and the associated abstract controller $\Conth$ can be computed as the nested iteration given in \REFalg{alg:compute wu}.
The controller $\Conth$ is a partial function from $\Xh$ to $\Us$, and we use the notation $\dom \Conth$ to denote the domain of the controller $\Conth$.

\begin{algorithm}
	\caption{Computation of $\wu$}
	\label{alg:compute wu}
	\begin{algorithmic}[1]
		\Require $\Bh \subseteq \Xh$
		\Ensure $\wu,\Conth$
		\State $Y \gets \Xh, Y' \gets\emptyset$ \label{line:initialize}
		\While{$Y \neq Y'$}\label{line:outer while begin}
			\State $Y' \gets Y$
			\State $Z \gets\emptyset, Z' \gets \Xh$
			\State $\forall\xh\in \Xh\;.\;\Conth:\xh \mapsto \emptyset$ 
			\While{$Z \neq Z'$}\label{line:inner while begin}
				\State $Z' \gets Z$
				\State $Z'' \gets (\Bh^c \cap (\apre(Y,Z)\cup \cpre_{\Fo}(Z))) \cup (\Bh\cap \cpre_{\Fo}(Y))$ \label{line:1 step wu}
				\State $\forall \xh\in Z''\setminus (\Bh \cup \dom \Conth)\;.\;\Conth:\xh\mapsto u$ s.t.\ $\Fo(\xh,u)\subseteq Y \wedge \Fu(\xh,u)\cap Z \neq \emptyset$ or $\Fo(\xh,u)\subseteq Z$ \label{line:control input for progress}
				\State $Z \gets Z''$
			\EndWhile\label{line:inner while end}
			\State $Y \gets Z$
		\EndWhile\label{line:outer while end}
		\State $\wu \gets Y$
		\State $\forall \xh\in \wu\cap \Bh\;.\;\Conth:\xh\mapsto u$ s.t.\ $\Fo(\xh,u)\subseteq \wu$\label{line:control input for safety}
		\State \Return $\wu,\Conth$
	\end{algorithmic}
\end{algorithm}

Note that, the existence of the control input $u$ in Lines~\ref{line:control input for progress} and~ \ref{line:control input for safety} 
is guaranteed because of the definition of $\cpre$ and $\apre$.


\begin{proposition}\label{prop:wu in w}
	The set $\wu$ is an under-approximation of the maximal a.s.\ winning region $\W(\Sys,\Cont^\ast)$.
\end{proposition}
\ifmaximal
	\begin{proof}
		The goal is to prove $ \wu\subseteq \W(\Sys,\Cont^\ast)$.
		Let $q\in \Xh$ be a state s.t.\ $q\in \wu$.
		We show that $q\subseteq \W(\Sys,\Cont^\ast)$.
		In the last iteration of the outer while loop in Alg.~\ref{alg:compute wu}, we obtain a growing sequence of states $Z_0 \subset Z_1 \subset \ldots \subset Z_k = Y$, where $Z_0 = \emptyset$ and $Y = \wu$.
		Since $\apre(Y,\emptyset) = \cpre_{\Fo}(\emptyset) = \emptyset$, hence $Z_1 \subseteq \Bh$.
		For all the other states $\xh$ in $Z_i$ for $i\in (2;k]$, one of the two cases happen: Either
		\begin{inparaenum}[(a)]
			\item by $\cpre_{\Fo}(Z_{i-1})$, it is ensured that $Z_{i-1}$ is surely reached from $\xh$ in one step, or
			\item by $\apre(Y,Z_{i-1})$, it is ensured that $Y$ (same as $\wu$ in the last iteration) is not left from $\xh$, and additionally (follows from \REFprop{prop:finite escape}) from all the (continuous) states inside $\xh$, transition to $Z_{i-1}$ happens almost surely in the long run.
		\end{inparaenum}
		Thus, for every $q\in \wu\setminus \Bh$ and for all $x\in q$, $\Bh$ is reached almost surely in the long run.
		
		Moreover, the operator $\cpre_{\Fo}(Y)$ ensures that $Y$---same as $\wu$ in the last iteration---is not left in the one step from $Z_1$.
		Hence almost surely $\Bh$ is visited infinitely often. 
	\end{proof}
\fi

\begin{remark}
	The operators $\apre$ and the fix-point \eqref{eq:wu} are inspired by how a.s.\ winning strategies are synthesized for B\"uchi specification in two-player concurrent games \cite{dAH00}:
	There the optimal strategy for the protagonist player (who wants to satisfy the B\"uchi condition) is to play an action that surely keeps the game within the winning region, while making progress towards the target with positive probability.
	
	At a very high level, we use the same insight to express $\wu$ in \eqref{eq:wu}, though for us the underlying game structure is totally different (see \REFrem{rem:two-adversaries}).
	It turns out that winning the game almost surely in our case means to either stay in the winning region using all $\Fo$-successors, while at the same time making progress using some $\Fu$-successor.
%
\end{remark}

\subsection{Controller refinement}
The state space discretization in the abstraction process induces a quantizer map $Q:\Xs \rightarrow \Xh'$ s.t.\ $Q:s\mapsto \xh$ when $s\in \xh$.
Given the abstract controller $\Conth:\Xh\rightarrow \Us$, we can obtain a continuous controller $\Cont:\Xs\rightarrow \Us$ as $\Cont \equiv \Conth\circ Q$, 
where ``$\circ$'' denotes function composition.
The following theorem states that $\Cont$ is a sound controller.

\begin{theorem}
Consider the control policy $\Cont \equiv \Conth\circ Q$ and the set $\wu$, where $\wu$ and $\Conth$ are the outputs of \REFalg{alg:compute wu} and $Q$ is the quantizer map. Then, $P_s^\Cont(\Sys \models \square\lozenge B) = 1$ for all $s\in \wu$.
\end{theorem}

The proof of the above theorem directly follows from the proof of \REFprop{prop:wu in w}, and hence is omitted.

\subsection{Over-approximation of the maximal a.s.\ winning region}
\label{sec:over_app}
The over-approximation of $\W(\Sys,\Cont^\ast)$ is given by the following fixed-point
\begin{multline}\label{eq:wo}
	\wo := \nu Y\;.\;\mu Z\;.\;\left[ (\Bho^c \cap \upre(Y,Z))\right. \\
									\left. \cup \left(\Bho \cap (\cpre_{\Fu}(Y) \cup \pre_{\Fo\setminus \Fu}(Y))\right) \right].
\end{multline}
The expression \eqref{eq:wo} can be solved in the same way as \REFalg{alg:compute wu} by replacing the update in Line~\ref{line:1 step wu} with the update in the r.h.s.\ of \eqref{eq:wo}.
Also, Line~\ref{line:control input for progress} and \ref{line:control input for safety} are not needed, as the control policy in this case does not serve any useful purpose.


\begin{proposition}\label{prop:w in wo}
	The set $\wo$ is a superset of the maximal a.s.\ winning region $\W(\Sys,\Cont^\ast)$.
\end{proposition}
\ifmaximal
	\begin{proof}
		Let $x^*\in \W(\Sys,\Cont^\ast)$.
		Since $\Xh$ creates a partition of the state space $\Xs$, hence there exists an abstract state $\xh^*\in \Xh$ s.t.\ $x^*\in \xh^*$.
		We need to show that $\xh^* \in \wo$.
		For the sake of contradiction, assume that $\xh^*\in \wo^c$.
		We will show that this cannot happen.
		
		The fixed point computation \eqref{eq:wo} produces a shrinking sequence of states $\Xh = Y^0 \supseteq Y^1 \supseteq \ldots \supseteq Y^k = \wo$.
		Let $i\in \mathbb{N}$ be the round index when $\xh^*$ was excluded from $Y$ for the first time i.e., $\xh^*\in Y^{i-1}$ but $\xh^*\notin Y^{i}$.
		Consider the following two possible cases:
		\begin{inparaenum}[(a)] 
			\item When $\xh^*\in \Bho$, then this means that for all $u\in \Us$, $\Fu(\xh^*,u) \not\subseteq Y^{i-1}$ (\emph{all states} in $\xh^*$ leave $Y^{i-1}$ with positive probability) and $\left(\Fo(\xh^*,u)\setminus \Fu(\xh^*,u)\right) \cap Y^{i-1} = \emptyset$ (\emph{no state} in $\xh^*$ can stay in $Y^{i-1}$ with positive probability).\label{case a}
			\item When $\xh^*\notin \Bho$, then this means that for all $u\in \Us$, either $\Fu(\xh^*,u) \not \subseteq Y^{i-1}$ (\emph{all states} in $\xh^*$ leave $Y^{i-1}$ with positive probability), or from \emph{all states} $x\in\xh^*$ there does not exist any path to $\Bho$. \label{case b}
		\end{inparaenum}
		Both \eqref{case a} and \eqref{case b} mean that from \emph{all} the continuous states $x\in\xh^*$, the specification will be violated with positive probability after $i$ time steps.
		This is a contradiction to our assumption that $x^*\in \W(\Sys,\Cont^\ast)$, since we know that from $x^*\in \xh^*$ the specification can be satisfied for infinite duration with probability $1$.
		Hence, it must hold that $\xh^*\in \wo$.	
	\end{proof} 
\fi

\subsection{Over-approximation of the minimal a.s.\ losing region for reachability}\label{sec:minimal losing region}
Once we have a tight approximation of the a.s.\ winning region, we can compute a lower-bound of the satisfaction probability for the quantitative version of the $\square\lozenge B$ through Eqn.~\eqref{eq:decom}:
\begin{equation*}
P_s^{\Cont}(\mathfrak S\models \square\lozenge B) \ge P_s^{\Cont}(\mathfrak S\models \lozenge \W(\Sys,\Cont^\ast))\ge P_s^\Cont(\Sys \models \lozenge \wu)
\end{equation*}
for all $s\in \wu^c$, 
where $\wu^c$ is the complement of $\wu$.
Efficient computation of maximal reachability requires computation of minimal a.s.\ losing region, i.e., the set
$$L:=\set{s\in\mathcal S | \sup_{\Cont\in\Pi_S}P_s^\Cont(\Sys \models \lozenge \wu) = 0}.$$
The following fixed point over-approximates $L$:
\begin{align*}
	\lo := \left[ \mu X\;.\; \pre_{\Fu}(X) \cup \wu \right]^c.
\end{align*}

\begin{theorem}
	$\lo$ is an over-approximation of $L$.
\end{theorem}
\ifmaximal
	\begin{proof}
		We show the contra-positive, i.e.\ $\lo^c\subseteq L^c $.
		Consider any abstract state $\xh\in \lo^c$.
		By construction, from all the continuous states $x\in \xh$, $\wu$ is reached with non-zero probability.
		Hence $\xh \notin L$.
	\end{proof}
\fi

Once we have an over-approximation of $L$, the stochastic kernel of the CMP becomes contractive over $\mathcal S\backslash(\lo\cup \wu)$ under mild continuity assumptions. Then approximate computational techniques in the literature on finite-horizon reachability can be utilised to find $P_s^\Cont(\Sys \models \lozenge \wu)$ with tunable error bounds \cite{SA13,SAH12,TMKA17}.

\begin{remark}\label{rem:wrong assumption}
	If \REFass{ass:markov policy} does not hold, then $\W(\Sys,\Cont^\ast)$ would be the 
	``largest a.s.\ winning region achievable using stationary policies'' contained in, but not necessarily equal to, the ``true a.s.\ winning region.''
	As a result, $\wo$ would only be an over-approximation of $\W(\Sys,\Cont^\ast)$, and not necessarily an over-approximation of the true a.s.\ winning region anymore.
	Nevertheless, the final controller obtained using \REFalg{alg:compute wu} would still be a sound a.s.\ winning controller, but with a smaller domain than the optimal a.s.\ winning controller.
	Moreover, $\lo$ would still be an over-approximation---but a more conservative one---of the ``true minimal a.s.\ losing region.''
\end{remark}

%% file: abstraction_computation.tex
\section{Computation of the abstraction}
\label{sec:abs computation}

\noindent\textbf{The dynamical system.}
We consider sampled-time continuous state dynamical systems with additive stochastic disturbance.
The system is formalized using the tuple $\Sigma = (\Xs,\mathcal U,f,t_w)$, where $\Xs \subset \mathbb{R}^n$ is the state space, $\mathcal U \subset \mathbb{R}^m$ is the finite input space, $f:\Xs\times \mathcal U \rightarrow \Xs$ is the nominal state transition function and $t_w:\mathbb{R}^n\rightarrow \mathbb R_{\ge 0}$ is the density function of the stochastic disturbance.
The state update of $\Sigma$ is given as:
\begin{equation}\label{eq:random traj}
	s(k+1) = f(s(k),u(k)) + w(k),\quad k\in\mathbb N, 
\end{equation}
where $s(k)\in \Xs$ and $u(k)\in \mathcal U$ are the state and input at the $k^{\text{th}}$ time instant, $w(k)$ is a random variable with the density function $t_w(\cdot)$, and $s(k+1)$ is the state at the $(k+1)^{\text{st}}$ time instant.

The random variables $\set{w(k)}_{k\in \mathbb{N}}$ are pairwise independent with the same density function $t_w(\cdot)$. We can write the system as a CMP $\Sys = \left(\Xs, \mathcal U, \Ker\right)$ with the stochastic kernel $\Ker(A\mid s,u) = \int_{A} t_w(s'-f(s,u))ds' $ for all $A\in \mathcal{B}(\mathcal S)$.
For the construction of the abstraction we assume that $t_w(\cdot)$ is piecewise continuous and $f(\cdot,u)$ is continuous for all $u\in\mathcal U$.

\medskip

\noindent\textbf{The abstraction.}
We assume that $\Xs = \Xs'\cup\{\phi\}$, where $\Xs'$ is a compact hyper-rectangular working region of the system and $\phi$ is a sink state representing the complement of $\Xs'$. The disturbance has a compact support $D\subset \mathbb{R}^n$. Let $\Xh'$ be a hyper-rectangular partition of $\Xs'$.
The overall abstract state space is $\Xh = \Xh'\cup \set{\phi}$.
Given an abstract state $\xh = \cell{a,b} \in \Xh'$ and a control input $u\in \mathcal U$, we denote the approximate \emph{nominal} reachable set of $\Sys$ by $\reach(\xh,u)$ s.t.\ 
\begin{equation}
\label{eq:psi}
\reach(\xh,u)\supseteq \bigcup_{s\in \mathsf{cl}(\xh)} f(s,u),
\end{equation}
where $\mathsf{cl}(\xh)$ is the closure of the set $\xh$.
Note that $\reach(\xh,u)$ can be computed using any reachability analysis method for deterministic dynamical systems \cite{coogan2015efficient,dang2012reachability}.\\
Define two functions $S_1,S_2:\Xh\times \mathcal U \rightarrow 2^{\mathbb R^n}$ s.t.\ 
\begin{align}
	&S_1:(\xh,u)\mapsto \overline{D}\oplus \reach(\xh,u)  \quad\text{and} \label{eq:S_1}\\
	&S_2:(\xh,u)\mapsto  \underline{D}\ominus (-\reach(\xh,u)), \label{eq:S_2}
\end{align}
where $\overline{D}\supseteq D$ is any over-approximation of the support of disturbance $D$ and 
$ \underline{D} \subseteq D$ is any compact under-approximation of $D$ over which $t_w(\cdot)$ is strictly positive.
The operators $\oplus$ and $\ominus$ are Minkowski sum and Minkowski difference of two sets, respectively, and the minus sign in $(-\reach(\xh,u))$ is applied to all elements. 
%
%
\begin{theorem}\label{thm:sound abstraction}
	Let $\Sigma = (\Xs,\mathcal U,f,t_w)$ be a dynamical system and $\Sys = \left(\mathcal S, \mathcal U, \Ker\right)$ be the CMP induced by $\Sigma$.
	Define $\Abs = (\Xh,\mathcal U,\overline{T},\underline{T})$ s.t.:
	\begin{align}
		&\overline{T}(\xh,u) := \set{\xh'\in \Xh \mid (\xh' \neq \phi)\Rightarrow (\xh'\cap S_1(\xh,u)\neq \emptyset ) \wedge\nonumber\\
		&\hspace{3.4cm}  (\xh' = \phi)\Rightarrow (S_1(\xh,u)\not\subseteq \Xs')},\label{eq:Tu}\\
		&\underline{T}(\xh,u) := \set{\xh'\in \Xh \mid (\xh' \neq \phi)\Rightarrow  (\lambda(\xh'\cap S_2(\xh,u))> 0) \wedge\nonumber\\
		&\hspace{3.4cm} (\xh' = \phi)\Rightarrow \lambda(S_2(\xh,u)\backslash\mathcal S')>0},
		\label{eq:To}
	\end{align}
	where $ \lambda(\cdot)$ gives the Lebesgue measure (volume) of a set. 
	Then $\Abs$ is an abstraction of $\Sys$.
\end{theorem}

\ifmaximal
	\begin{proof}
		We show that $\overline{T}$ and $\underline{T}$ satisfy the properties of $\Fo$ and $\Fu$ as formalized in \REFdef{def:abstraction}.
		For \eqref{eq:def fo}, consider any pair of abstract states $\xh,\xh'\in \Xh'$ and input $u\in U$ and there exists $s\in \xh$, $\Ker(\xh'\mid s,u)>0$. We show that $\xh'\in \overline{T}(\xh,u)$:
		\begin{align*}
		 & \int_{\xh'} t_w(s'-f(s,u))ds'>0\Rightarrow\int_{\xh'\ominus\set{f(s,u)}} t_w(w)dw>0\\
		 & \Rightarrow (\xh'\ominus\set{f(s,u)})\cap D\neq\emptyset\\
		 & \Rightarrow  \exists w\in D,  \exists s'\in\xh' \text{ s.t.}\quad w=s'-f(s,u)\\
		 & \Rightarrow  \exists w\in D,  \exists s'\in\xh' \text{ s.t.}\quad s'=f(s,u)+w.
		\end{align*}
		At the same time we know that $f(s,u)\in \reach(\xh,u)$ since $s\in \xh$. Then,
		\begin{align*}
		 & s'\in S_1(\xh,u)\Rightarrow \xh'\cap S_1(\xh,u)\neq\emptyset
		 \Rightarrow \xh'\in \overline{T}(\xh,u).
		 \end{align*}
		 A similar reasoning holds for the case of $\xh'=\phi$.\\
		Now we show that $\underline{T}$ satisfies the condition given in \eqref{eq:def fu}. Take $\xh\in \Xh'$, input $u\in U$, and $\xh'\in \underline{T}(\xh,u)$ s.t.\ $\xh'\neq\emptyset$. Then $\lambda(\xh'\cap S_2(\xh,u))> 0$ according to \eqref{eq:To}. For any $s'\in \xh'\cap S_2(\xh,u)$, we have
		\begin{align*}
		 & s'\in S_2(\xh,u)\Rightarrow \set{s'}\oplus(-\reach(\xh,u))\subseteq\underline{D}\\
		 & \Rightarrow s'-f(s,u) \subseteq\underline{D}\quad \forall s\in\mathsf{cl}(\xh)\\
		 & \Rightarrow \Ker(\xh'\mid s,u) \ge \int_{\xh'\cap S_2(\xh,u)} t_w(s'-f(s,u))ds'>0
		 \end{align*}
		The right-hand side is strictly positive since the integrand is strictly positive and the domain of integration has a positive measure. It is also assumed that $f$ is continuous and $t_w$ piecewise continuous. Therefore, we have a positive function over the compact domain $\mathsf{cl}(\xh)$, which will have a positive minimum:   
	$$\exists \varepsilon>0\;.\; \forall s\in \mathsf{cl}(\xh)\;.\;\Ker(\xh'\mid s,u)\ge\varepsilon \;\;\Rightarrow\;\;\xh'\in \Fu(\xh,u).$$
	\end{proof}
\fi
	
	The abstraction procedure can be summarized as follows: first compute the approximate nominal reachable set $\reach(\xh,u)$ in \eqref{eq:psi}, then take the Minkowski sum and difference for $S_1,S_2$ in \eqref{eq:S_1}-\eqref{eq:S_2}, and finally compute the transition relations \eqref{eq:To}-\eqref{eq:Tu}. \REFfig{fig:abs computation} illustrates the abstraction procedure for a 2-d system and when $D$ is of the form $[-d_1,d_1]\times [-d_2,d_2]$.  
	
	\begin{figure}
	\begin{tikzpicture}[scale=0.7]
		\newlength{\dmaxA};
		\setlength{\dmaxA}{1.2cm};
		\newlength{\dmaxB};
		\setlength{\dmaxB}{1cm};
		\coordinate (A)	at	(0,0);
		\newlength{\lA};
		\setlength{\lA}{1cm};
		\newlength{\wA};
		\setlength{\wA}{0.8cm};
		\draw[fill=gray!50]	 (A)	rectangle	($(A)+(\lA,\wA)$)	node[pos=0.5]	{$\xh$};
		\draw[xstep=\lA,ystep=\wA,gray!50,thin] (-1,-2) grid (6,3);		
		
		\coordinate (B)	at	(2.9,-0.3);
		\newlength{\lB};
		\setlength{\lB}{1.8cm};
		\newlength{\wB};
		\setlength{\wB}{1.5cm};
		\draw[fill=yellow,opacity=0.2]	 ($(A)+(\lA,-2*\wA)$)	rectangle	($(A)+(6*\lA,3*\wA)$);
		\draw[fill=green,opacity=0.2]	($(A)+(3*\lA,0)$)		rectangle	($(A)+(5*\lA,\wA)$);
		\draw	[thick] (B)	rectangle	($(B)+(\lB,\wB)$);

		\draw[fill=red,opacity=0.1, dashed]	($(B)-(\dmaxA,\dmaxB)$)	rectangle	($(B)+(\dmaxA,\dmaxB)$);
		\draw[fill=red,opacity=0.1,dashed]	($(B)+(\lB-\dmaxA,\wB-\dmaxB)$)	rectangle	($(B)+(\lB+\dmaxA,\wB+\dmaxB)$);
		
		\draw[thick,blue] ($(B)-(\dmaxA,\dmaxB)$)	rectangle	($(B)+(\lB+\dmaxA,\wB+\dmaxB)$);
		
		\draw[thick,red] ($(B)+(\lB-\dmaxA,\wB-\dmaxB)$)	rectangle	($(B)+(\dmaxA,\dmaxB)$);
		
		\draw[<->]	($(B)+(-\dmaxA,0)$)	--	node[above]	{$d_1$}	($(B)+(0,0)$);
		\draw[<->]	($(B)+(\lB,\wB)$)	--	node[above]	{$d_1$}	($(B)+(\lB+\dmaxA,\wB)$);
		\draw[<->]	($(B)+(\lB,\wB)$)	--	node[left]	{$d_2$}	($(B)+(\lB,\wB+\dmaxB)$);
		\draw[<->]	($(B)+(0,0)$)	--	node[left]	{$d_2$}	($(B)+(0,-\dmaxB)$);
		
%
%
%
		
	\end{tikzpicture}
	\caption{Illustration of abstraction computation: given the abstract state $\xh$ (filled with grey) and some control input $u$, first the nominal reachable set is over-approximated (black rectangle). Next, the sets $S_1$ (blue rectangle) and $S_2$ (red rectangle) are computed. Finally, the images of the transition functions $\Fo$ (filled with yellow) and $\Fu$ (filled with green) are the abstract states intersecting with $S_1$ and $S_2$ respectively.}
	\label{fig:abs computation}
\end{figure}
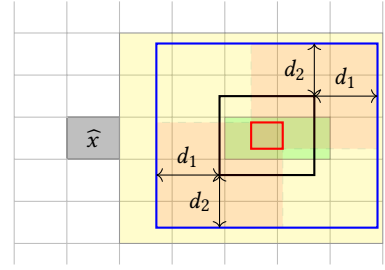

\subsection{Computation for mixed-monotone systems}
If the function $f(\cdot,u)$ is mixed-monotone for every $u\in U$ and the partition sets are hyper-rectangles, then the nominal reachable set $\reach(\xh,u)$ can be computed particularly efficiently.
We recall the definition of mixed-monotonicity \cite{coogan2015efficient}.
\begin{definition}
	Let $g:\Xs\rightarrow \Xs$ be a function, and $\leq_\Xs$ be an order relation on $\Xs$ induced by positive cones.
 	The function $g$ is called mixed-monotone w.r.t.\ $\leq_\Xs$ (or simply mixed-monotone if $\leq_\Xs$ is obvious from the context) if there exists a function $h:\Xs\times \Xs\rightarrow\Xs$---called the decomposition function---with the following properties:
	\begin{enumerate}
		\item $\forall x\in \Xs\;.\; h(x,x) = g(x)$,
		\item $\forall x_1,x_2,y\in \Xs\;.\;(x_1\leq_\Xs x_2) \Rightarrow (h(x_1,y)\leq_\Xs h(x_2,y))$, and
		\item $\forall x,y_1,y_2\in \Xs\;.\;(y_1\leq_\Xs y_2) \Rightarrow (h(x,y_2)\leq_\Xs h(x,y_1))$.
	\end{enumerate}
\end{definition}

Intuitively, a mixed-monotone function can be decomposed into an increasing and a decreasing component.
This phenomenon can be seen from the definition of the decomposition function.
The following proposition shows a fast over-approximation method of the image of a rectangular set under a mixed-monotone function.
\begin{proposition}[{\cite[Thm.~1]{coogan2015efficient}}]\label{prop:approx reach set mixed-monotone}
	Let $g$ be a mixed-monotone function with the decomposition function $h$, and $\cell{a,b}\subseteq \Xs$ be any hyper-rectangle.
	The image of $\cell{a,b}$ under $g$ can be over-approximated as $\cell{h(a,b),h(b,a)}$.
\end{proposition}

For mixed-monotone $f(\cdot,u)$ with decomposition function $h_u$, the function $\reach(\xh,u)$ 
can be computed using \REFprop{prop:approx reach set mixed-monotone} as $\reach(\xh,u) = \cell{h_u(a,b),h_u(b,a)}$ for $\xh = \cell{a,b}$.

%% file: examples.tex
\section{Examples}

We implemented the presented algorithms in an open-source tool called Mascot-SDS (\url{https://gitlab.mpi-sws.org/kmallik/mascot-sds.git}), which is an extension of the tool Mascot \cite{hsu2018multi}, and is built on top of the basic symbolic computation framework of the tool SCOTS \cite{SCOTS}.
All the experiments were performed on a computer with a 3GHz Intel Xeon E7-8857 v2 processor and 1.5 TB memory.

\subsection{Perturbed Van der Pol Oscillator}

Our first example considers the computation of the maximal a.s.\ winning region of an autonomous stochastically perturbed Van der Pol oscillator \cite{Pol}.
The state evolution of the oscillator is given by:
\begin{align*}
	x_1(k+1) &= x_1(k) + x_2(k)\tau + w_1(k)\\
	x_2(k+1) &= x_2(k) + (-x_1(k)+(1-x_1(k)^2)x_2(k))\tau + w_2(k),
\end{align*}
where the sampling time $\tau$ is set to $0.1s$ and $(w_1(k),w_2(k))$ is a pair of stochastic noise signals at time $k$ drawn from a piecewise continuous density function with a compact support $D = [-0.02,0.02]\times [-0.02,0.02]$. 
Note that for the computation of the winning set, we do not need the actual density function as discussed in the previous section.
We consider a safety specification $\square\Xs'$ for staying within the working area $\Xs'=[-0.5,0.5]\times [-0.5,0.5]$, as well as a B\"uchi specification $\square\lozenge B$ for repeatedly reaching the target set $B=[-1.2,-0.9]\times [-2.9,-2]$ (green rectangle in \REFfig{fig:vanderpol}).
Our algorithm is able to compute the under and over-approximation of the set of a.s.\ winning region.
In \REFfig{fig:vanderpol}, the under-approximation $\wu$ is shown in grey and $\wo\setminus \wu$ is shown in blue.

It turns out that when the noise is treated as worst case, then there exists a deterministic value of the noise for which the oscillator trajectory \emph{never} reaches the target from all the initial states inside the domain, thus violating the specification.
So the winning region is empty. 
A trajectory with a fixed deterministic perturbation that misses the target all the time is shown in black in \REFfig{fig:vanderpol}.

On the other hand, when the noise is treated as stochastic, then there are initial states from where the perturbed trajectory always eventually reaches the target polytope.
Hence, the specification is satisfied. A trajectory with stochastic perturbation and the initial state $I$  is shown in red in \REFfig{fig:vanderpol}.
\REFtab{table:performance summary} summarizes the abstraction parameters used in our experiment, ratio of the computed volume of $\wu$ to the computed volume of $\wo$, and the computation time.

%

\subsection{Controlled perturbed vehicle}

Our second example is a controller synthesis problem for a perturbed sampled-time 3-d Dubins' vehicle \cite{Vehicle}.
We consider almost sure satisfaction of the B\"uchi specification while avoiding obstacles in the state space.
Although we did not discuss avoidance of obstacle in the theory part, this can be easily handled by redefining the working region $\Xs'$ of the system by excluding the obstacles.
Thus given a hyper-rectangular working region $[0,2]\times [0,3]\times [-\pi,\pi] $, and an obstacle $[0.8,1.2]\times [1,1.4]\times [-\pi,\pi]$ within that working region, we define  $\Xs'=[0,2]\times [0,3]\times [-\pi,\pi] \setminus [0.8,1.2]\times [1,1.4]\times [-\pi,\pi]$.
The system dynamics is given as: when $u\neq 0$,
\begin{align*}
	&x_1(k+1) = x_1(k) + \frac{V}{u}\sin(x_3(k) + u\tau) - \frac{V}{u}\sin(x_3(k)) + w_1(k)\\
	&x_2(k+1) = x_2(k) - \frac{V}{u}\cos(x_3(k) + u\tau) + \frac{V}{u}\cos(x_3(k)) + w_2(k)\\
	&x_3(k+1) = x_3(k) + u\tau + w_3(k),
\end{align*}
and when $u = 0$,
\begin{align*}
	&x_1(k+1) = x_1(k) + V\cos(x_3(k))\tau + w_1(k)\\
	&x_2(k+1) = x_2(k) - V\sin(x_3(k))\tau + w_2(k)\\
	&x_3(k+1) = x_3(k) + w_3(k),
\end{align*}
where the sampling time $\tau = 1s$, the constant forward velocity $V=0.1$ (maintained by e.g.\  a low level cruise control system), and $(w_1(k),w_2(k),w_3(k))$ is a collection of stochastic noise samples drawn from a piecewise continuous density function with the support $D = [-0.06,0.06]\times [-0.06,0.06]\times [-0.06,0.06]$.
It is due to this fixed velocity that the vehicle cannot stay stationary (or near stationary) after reaching the target, which makes the synthesis problem with B\"uchi specification much more challenging than the same with normal reachability specification.

When the noise is treated as a worst case adversary, the winning region is empty.
However, when the noise is treated as stochastic, the approximate winning regions $\wu$ and $\wo$ are non-empty as shown in \REFfig{fig:vehicle domain}.
\REFfig{fig:vehicle simulation} shows the simulated trajectory of the vehicle using the synthesized controller.
It was observed that even though the trajectory moves away from the target from time to time, either due to the external noise or due to the constant velocity, it always returns to the target eventually.

We performed the computation for four different levels of discretization granularity, and the results are summarized in \REFtab{table:performance summary}.
It can be observed (from the ratio $\lambda(\wu)/\lambda(\wo)$) that the gap between $\wo$ and $\wu$ monotonically shrinks as the abstract states get smaller, which means that the approximations $\wo$ and $\wu$ get progressively better with refinement of the state space partition.
However, computation time sharply increases with finer discretization.

\ifmaximal 
\subsection{A note on computation time}
\fi
The computation time for $\wu$ reported in \REFtab{table:performance summary} is based on a warm-start of \REFalg{alg:compute wu} by replacing $Y\gets \Xh$ in Line~\ref{line:initialize} with $Y\gets\wo$.
The intuition is that since it is known upfront that $\wu\subseteq \wo$, hence we do not need to consider the set $\Xh\setminus \wo$ in \REFalg{alg:compute wu}.
In practice, the computation time for $\wu$ would be higher than the numbers reported in \REFtab{table:performance summary} had we started with $\Xh$.

\ifmaximal
	In general, we observed that the computation of $\wu$ takes much longer than the computation of $\wo$.
	Our hypothesis is that this is due to the properties of the operators defined in \eqref{eq:cpre}-\eqref{eq:upre}, and how they are used in the computation of $\wu$ and $\wo$.
	For example, because $\Fo(\xh,u)$ is a superset of $\Fu(\xh,u)$ for all $\xh,u$, it can be shown that for a given $Y,Z\subseteq \Xh$, $\apre(Y,Z) \subseteq \upre(Y,Z)$.
	Similarly $\cpre_{\Fo}(Z)\subseteq \upre(Y,Z)$ when $Y\supseteq Z$.
	Moreover, $\Bh\cap \cpre_{\Fo}(Y) \subseteq \Bho \cap (\cpre_{\Fu}(Y) \cup \pre_{\Fo\setminus \Fu}(Y))$.
	Thus each iteration in the inner ``$\mu$'' fixed point would add possibly fewer states in case of $\wu$ than in case of $\wo$.
	Since ultimately the size of $\wo$ and $\wu$ are not very far apart, as shown in Col.~4 of \REFtab{table:performance summary}, hence the iterations for $\wu$ would take many more number of steps than $\wo$.
\fi

\ifmaximal
	\begin{figure*}
\else
	\begin{figure*}[t]
\fi
	\centering
	\subfloat[][]{
		\includegraphics[scale=0.29]{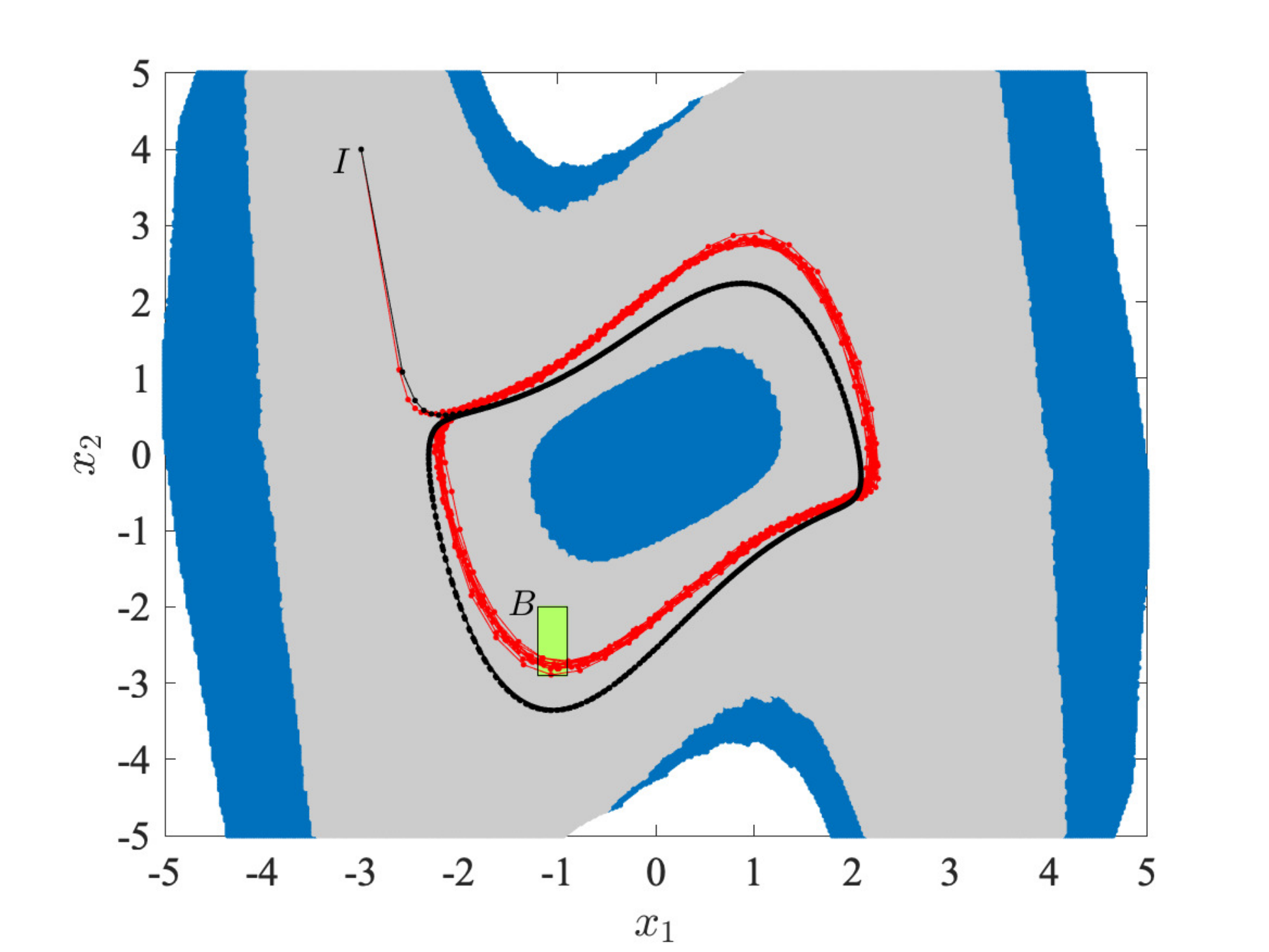}
		\label{fig:vanderpol}
	}
	\subfloat[][]{
		\includegraphics[scale=0.29]{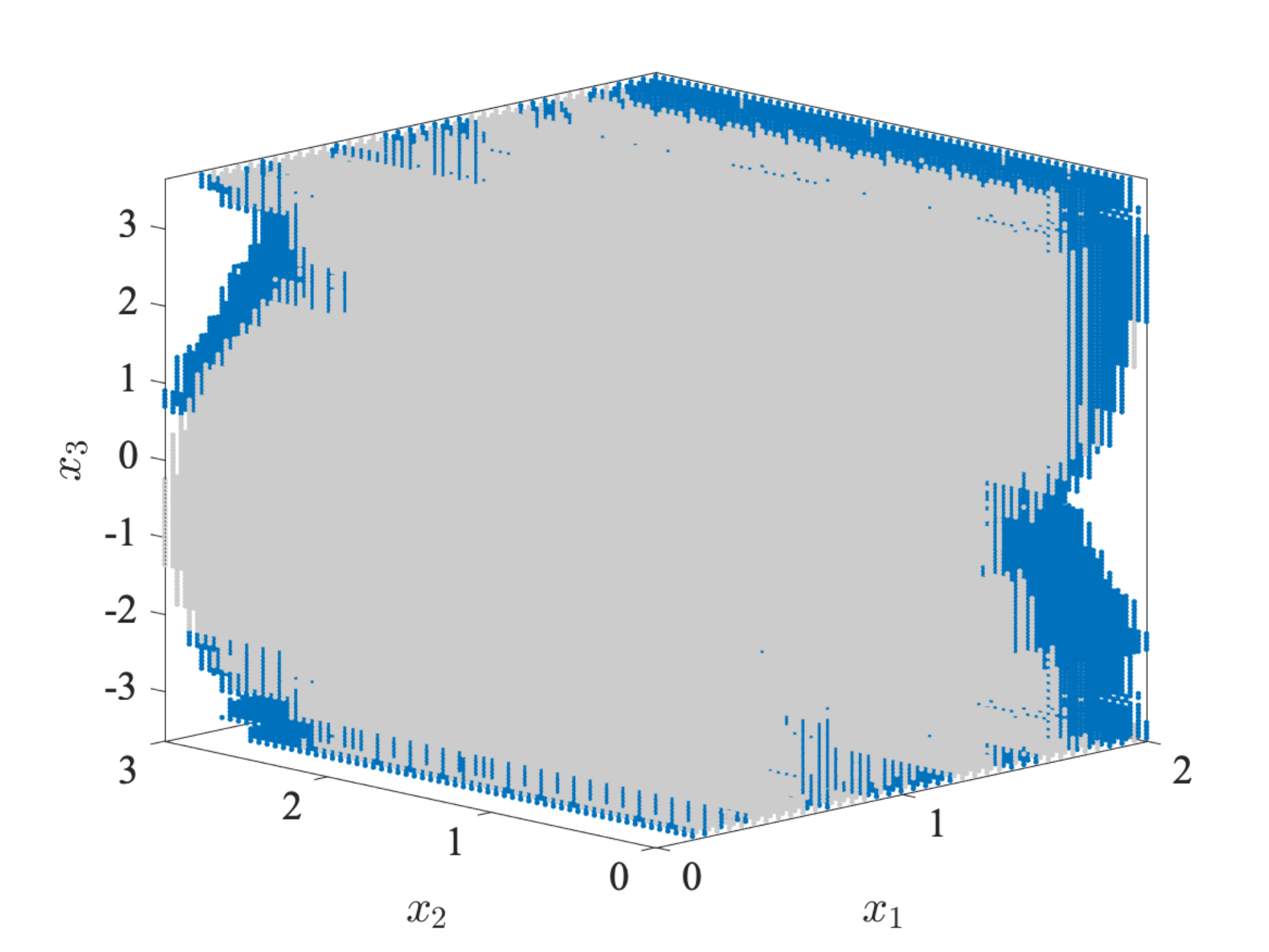}
		\label{fig:vehicle domain}
	}
	\subfloat[][]{
		\includegraphics[scale=0.29]{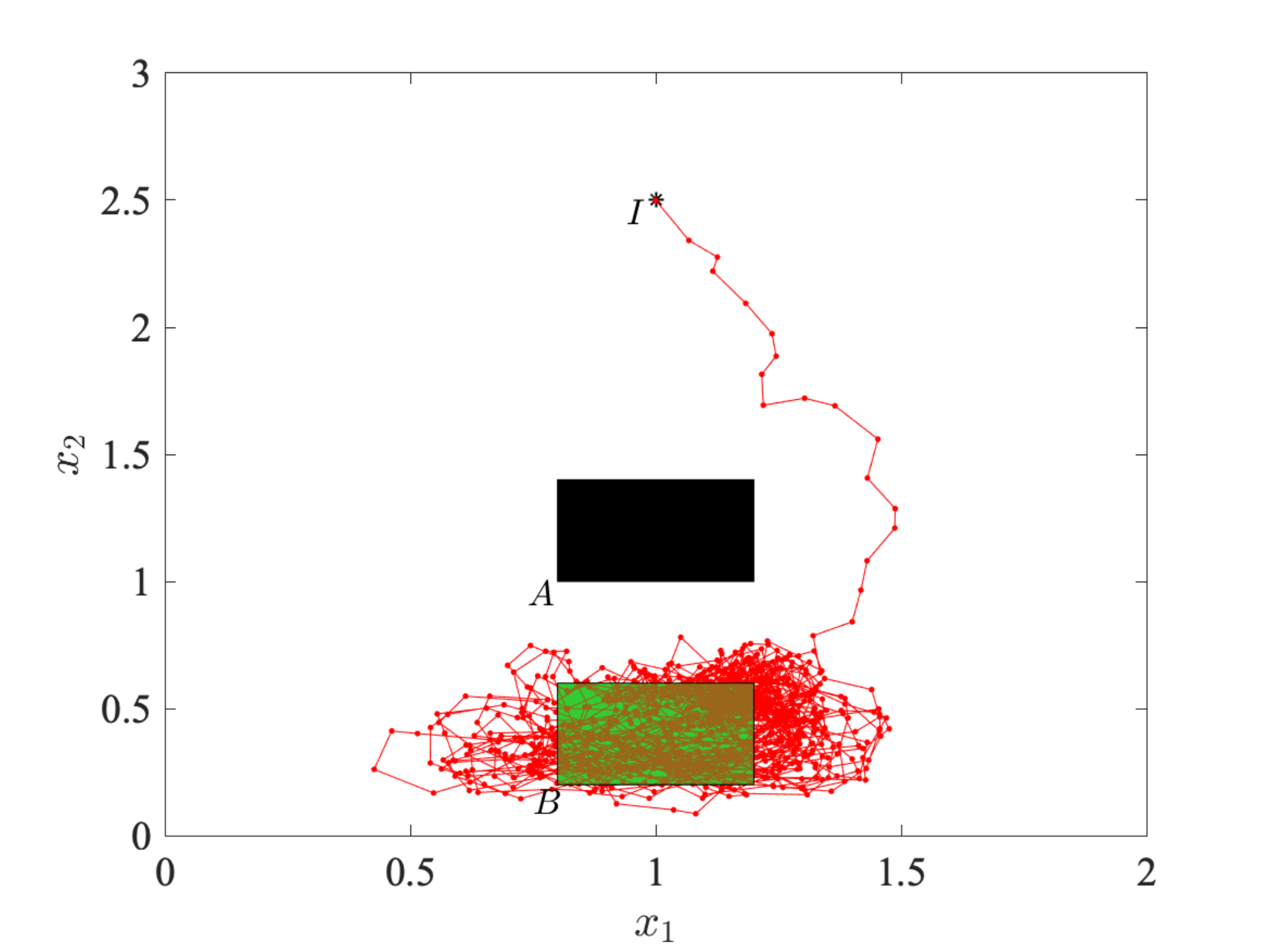}
		\label{fig:vehicle simulation}
	}
	\caption{\protect\subref{fig:vanderpol} The Van der Pol oscillator example:
		$B$ (green box) is the target, $\wu$ is in grey, $\wo\setminus \wu$ is in blue, and $I$ is the initial state for simulation.
		The trajectory with stochastic perturbation is shown in red, and the trajectory with a fixed deterministic perturbation that always misses the target is shown in black.  \protect\subref{fig:vehicle domain} The approximate a.s.\ winning domain for the Dubin's vehicle (for abstract state size $0.1\times 0.1\times 0.1$): $\wu$ is in grey and $\wo\setminus \wu$ is in blue. \protect\subref{fig:vehicle simulation} Simulation of perturbed trajectory for the Dubin's vehicle projected onto $x_1,x_2$ plane. The black box $A$ is the obstacle and the green box $B$ is the target.}
	\label{fig:all experiments}
\end{figure*}

\ifmaximal
	\begin{table*}
\else
	\begin{table*}[t]
\fi
	\begin{tabular}{|>{\centering\arraybackslash}m{3cm}|>{\centering\arraybackslash}m{3cm}|>{\centering\arraybackslash}m{2.2cm}|>{\centering\arraybackslash}m{2cm}|>{\centering\arraybackslash}m{2cm}|>{\centering\arraybackslash}m{1.5cm}|>{\centering\arraybackslash}m{1.5cm}|>{\centering\arraybackslash}m{1.5cm}|}
		\hline
		\multirow{3}{*}{System} & \multirow{3}{*}{\parbox{3cm}{\centering Size of abstract states}} & \multirow{3}{*}{\parbox{2.2cm}{\centering Volume of the worst-case winning region}} & \multirow{3}{*}{$\lambda(\wu)/\lambda(\wo)$} & \multicolumn{3}{c|}{Computation time}\\
		\cline{5-7}
		& & & & \multirow{2}{*}{Abstraction} & \multirow{2}{*}{$\wo$} & \multirow{2}{*}{$\wu$}  \\
		 & & & &  &  & \\
		 \hline 
		Van der pol oscillator & $0.02 \times 0.02$ & empty	& $73.0\%$ & $<1m$ & $12m$  & $416m$ \\
		\hline
		\multirow{4}{*}{Dubins' vehicle} & $0.1\times 0.1\times 0.1$ & empty & $80.0\%$ & $1m$ & $2m$ & $11m$ \\
		 & $0.07\times 0.07\times 0.07$ & empty & $83.8\%$ & $4m$ & $23m$ & $35m$ \\
		 & $0.05\times 0.05\times 0.05$ & empty & $87.2\%$ & $18m$ & $124m$ & $140m$ \\
		 & $0.04\times 0.04\times 0.04$ & empty & $88.9\%$ & $37m$ & $132m$ & $128m$ \\
		 \hline
	\end{tabular}
	\caption{Performance evaluation of our method on the Van der Pol oscillator and the Dubin's vehicle. The 2nd column shows the size of each hyper-rectangular abstract state in the underlying uniform grid, the 3rd column shows the volume of the approximate winning domain when the noise is treated in the usual worst-case sense, the 4th column shows the ratio of the Lebesgue measure (volume) of $\wu$ to $\wo$, and the 5th, 6th, and 7th columns show the computation times of different phases of our algorithm in minutes. Note that in our implementation, the computation of $\wu$ was warm started with already computed $\wo$. Had $\wu$ been computed from scratch, the computation time for $\wu$ would be higher than what is shown in the last column.}
	\label{table:performance summary}
\end{table*}



%% file: appendix.tex

\section{Appendix}

\subsection{Proof of \REFprop{prop:measurability of windom}}
Following the steps utilized in \cite[Theorem 7]{TMKA17}, we have that $1-f^\ast(s)$ is lower semi-analytic. Then $\{s\in\mathcal S\,|\, 1-f^\ast(s)<c\}$ is an analytic subset of $\mathcal S$ for all $c\in\mathbb R$. Take a positive sequence $\{c_n\rightarrow 0\}$. The set $\cap_n \{s\in\mathcal S\,|\, f^\ast(s)>1-c_n\} = \{s\in\mathcal S\,|\, f^\ast(s)=1\}$ is also analytic. Every analytic set is universally measurable.

\subsection{Proof of \REFthm{thm:the problem can be broken down}}

	\begin{proof}[Proof of \REFthm{thm:the problem can be broken down}]
	We already know that $P_s^{\Cont}(\mathfrak S\models \square\lozenge B) = 1$ for all $s\in \W(\Sys,\Cont)$ by definition of the winning set. Take any $s\notin W:= \W(\Sys,\Cont)$.
	We make the event conditional on $\tau$ which is the first time the path hits a state in $W$. Then we have
	\begin{align*}
	P_s^{\Cont}(\mathfrak S\models \square\lozenge B) & = \mathbb E_s^{\Cont}\left[P_s^{\Cont}(\mathfrak S\models \square\lozenge B\,|\, s_1,s_2,\ldots,s_n, \tau=n)\right]\\
	& = \sum_{n=0}^{\infty} P_s^{\Cont}(s_1,s_2,\ldots,s_{n-1}\in\mathcal S\backslash W, s_n\in W)\\
	& + P_s^{\Cont}(\mathfrak S\models \square\lozenge B \text{ and } \mathfrak S\models \square \mathcal S\backslash W) .
	\end{align*}
	The sum is the reachability probability and the last term is always non-negative.
	\end{proof}
	

\subsection{Proof of \REFprop{prop:cont}}
\begin{proof}
We take an element $\xh'\in\Xh$ and $u\in\mathcal U$. We show that $\xh'$ belongs to the right-hand side of \eqref{eq:under_alter} if and only if it belongs to the right-hand side of \eqref{eq:def fu}. Observe that
\begin{align*}
\Ker&(\xh'\mid s,u) = \Ker(\mathsf{cl}(\xh')\mid s,u) \\
&= \int_{\mathsf{cl}(\xh')} \Ker(ds'\mid s,u)
 = \int_{\Xs}\mathbf 1_{\mathsf{cl}(\xh')}(s') \Ker(ds'\mid s,u).
\end{align*}
Since the indicator function of a closed set is upper semi-continuous, $ \Ker(\xh'\mid s,u)$ is also upper semi-continuous for all $s\in\mathsf{cl}(\xh)$. Similarly, 
\begin{align*}
\Ker& (\xh'\mid s,u) = \Ker(\mathsf{int}(\xh')\mid s,u)\\
&  = \int_{\mathsf{int}(\xh')} \Ker(ds'\mid s,u)  = 
\int_{\Xs}\mathbf 1_{\mathsf{int}(\xh')}(s') \Ker(ds'\mid s,u).
\end{align*}
Since the indicator function of an open set is lower semi-continuous, $ \Ker(\xh'\mid s,u)$ is also lower semi-continuous for all $s\in\mathsf{cl}(\xh)$.
Therefore, $ \Ker(\xh'\mid s,u)$ is continuous over the domain $s\in\mathsf{cl}(\xh)$, which means it attains its minimum over $s\in\mathsf{cl}(\xh)$. Suppose $\xh'$ belongs to the right-hand side of \eqref{eq:under_alter} and take $\varepsilon = \min_{s\in\mathsf{cl}(\xh)}\Ker(\xh'\mid s,u)$. Since $\Ker(\xh'\mid s,u)$ is positive, $\varepsilon>0$ and $\xh'$ belongs to the right-hand side of \eqref{eq:def fu}. Now suppose $\xh'$ belongs to the right-hand side of \eqref{eq:def fu}. Then there is an $\varepsilon>0$ such that $\Ker(\xh'\mid s,u)\ge \varepsilon$ for all $s\in\xh$. Since $\Ker(\xh'\mid s,u)$ is continuous over $\mathsf{cl}(\xh)$, we get that $\Ker(\xh'\mid s,u)>0$ for all $s\in\mathsf{cl}(\xh)$, which means $\xh'$ belongs to the right-hand side of~\eqref{eq:under_alter}.
\end{proof}

\subsection{Properties of the winning region}
\label{sec:unique winning region}

\begin{proposition}
For any control policy $\Cont$, The set $W := \W(\Sys,\Cont)$ is an absorbing set, i.e., the paths starting from this set will stay in the set a.s.:
\[ P_s^{\Cont}(w[k+1]\in W \,|\, w[k]) = 1,\]
for all $w[k]\in W$.
\end{proposition}
\begin{proof}
For any $s\in W$, we have
\begin{align*}
&P_s^{\Cont}(\mathfrak S\models \square\lozenge B)
= \int_{\mathcal S}P_{s_1}^{\Cont}(\mathfrak S\models \square\lozenge B)\Ker(ds_1|s,\Cont(s))\\
&=  \int_{W} \Ker(ds_1|s,\Cont(s)) + \int_{\mathcal S\backslash W} P_{s_1}^{\Cont}(\mathfrak S\models \square\lozenge B)\Ker(ds_1|s,\Cont(s)).
\end{align*}
This means 
\begin{align*}
&\int_{\mathcal S\backslash W} (1-P_{s_1}^{\Cont}(\mathfrak S\models \square\lozenge B))\Ker(ds_1|s,\Cont(s))=0,\\
&P_s^{\Cont}\left[(1-P_{s_1}^{\Cont}(\mathfrak S\models \square\lozenge B))\mathbf{1}_{\mathcal S\backslash W}(s_1)\ge \epsilon\right]\le \frac{0}{\epsilon} = 0,
\end{align*}
where the last inequality is a consequence of Markov's inequality for non-negative random variables. By taking the union over a monotone positive sequence $\set{\epsilon_n\rightarrow 0}$, we get
\begin{align*}
&P_s^{\Cont}\left[(1-P_{s_1}^{\Cont}(\mathfrak S\models \square\lozenge B))\mathbf{1}_{\mathcal S\backslash W}(s_1)>0 \right] = 0,\\
&P_s^{\Cont}\left[s_1\in \mathcal S\backslash W \text{ and } P_{s_1}^{\Cont}(\mathfrak S\models \square\lozenge B)<1\right] = 0,\\
&P_s^{\Cont}\left[ s_1\in \mathcal S\backslash W\right] = 0.
\end{align*}
\end{proof}

\begin{proposition}
Given a countable sequence of stationary policies $\set{\Cont_1,\Cont_2,\ldots}$ for the system $\Sys$ with winning regions $\set{\W(\Sys,\Cont_n),\,n=1,2,\ldots}$, there is a controller $\Cont$ with winning region \ $\W(\Sys,\Cont) = \cup_{n=1}^{\infty} \W(\Sys,\Cont_n)$. 
\end{proposition}
\begin{proof}
Define the sets $\set{\overline{W}_n,n=1,2,\ldots}$ inductively as $\overline{W}_1 := \W(\Sys,\Cont_1)$ and $\overline{W}_n := \W(\Sys,\Cont_n)\backslash\cup_{i=1}^{n-1}\overline{W}_n$ for all $n\in\set{2,3,\ldots}$. 
This construction is illustrated in \REFfig{fig:venn for overall windomain}.
Also, define the new stationary policy:
\begin{equation}
\Cont(s) :=
\begin{cases}
\Cont_1(s) & \text{ if } s\in \overline{W}_1\\
 \Cont_2(s) & \text{ if } s\in \overline{W}_2\\
 \vdots
\end{cases} 
\end{equation}
for all non-empty sets $\overline{W}_n$. 
It is easy to show that the sets $\set{\overline{W}_i}$ are non-intersecting and $\cup_{i=1}^{n}\overline{W}_i = \cup_{i=1}^{n} \W(\Sys,\Cont_i)$. 
Then for any initial state $s\in \cup_{i=1}^{\infty} \W(\Sys,\Cont_i)$, there is some $n$ such that $s\in \overline{W}_n$. Note that all sets $\W(\Sys,\Cont_i)$ are absorbing under their respective policy. 
The path starting from $s$ with $\Cont_n$ either stay in $\W(\Sys,\Cont_n)$ or will reach some $\overline{W}_i$ with $i<n$. 
In the first case, the path satisfies the specification with probability one. The same argument can be applied a finite number of times until reaching the lowest index $i=1$.

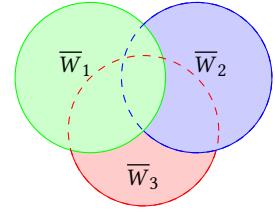
\begin{figure}
	\input{Figures/controller_union}
	\caption{Illustration of construction of $\overline{W}_1$ (green filled part), $\overline{W}_2$ (blue filled part), and $\overline{W}_3$ (red filled part) from $\W(\Sys,\Cont_1)$ (green circle), $\W(\Sys,\Cont_2)$ (blue circle), and $\W(\Sys,\Cont_3)$ (red circle).}
	\label{fig:venn for overall windomain}
\end{figure}

\textbf{measurability of $\Cont$.} Note that the sets $\set{\overline{W_1}, \overline{W_2},\ldots}$ are universally measurable and the policies $\set{\Cont_1,\Cont_2,\ldots}$ are universally measurable functions. We also have $\Cont^{-1}(A) = \cup_{i=1}^{\infty} [\Cont_i^{-1}(A)\cap \overline{W_i}]$, which means $\Cont^{-1}(A)$ is universally measurable for any universally measurable $A$. Therefore, $\Cont$ is a universally measurable function.


\end{proof}

%% file: Figures/controller_union.tex
\begin{tikzpicture}
	\draw[red,fill=red!20,fill opacity=1] (0,0)		circle		(1cm);
	\draw[blue,fill=blue!20,fill opacity=1] (0.707,0.707)		circle		(1cm);
	\draw[green,fill=green!20,fill opacity=1]	 (-0.707,0.707)	circle		(1cm);
	
	\draw[blue, dashed]	(0.707,0.707)		circle		(1cm);
	\draw[red,dashed] (0,0)		circle		(1cm);
	
	\node	at	(-0.9,0.9)	{$\overline{W}_1$};
	\node	at	(0.9,0.9)		{$\overline{W}_2$};
	\node	at	(0,-0.6)	{$\overline{W}_3$};
	
\end{tikzpicture}